\documentclass[format = sigplan, 9pt]{acmart}

\pdfoutput=1

\usepackage[T1]{fontenc}
\usepackage[utf8]{inputenc}
\usepackage{pgfplots}
\usepackage{enumerate,paralist}
\usepackage{pbox}

\usepackage{comment}
\usepackage{parskip}
\usepackage{multirow}
\usepackage{url}
\usepackage{graphicx}
\usepackage{alltt}

\usepackage{tikz,pgffor}
\usepackage{tkz-euclide}
\usetkzobj{all}
\usetikzlibrary{arrows}
\usetikzlibrary{shapes}
\usetikzlibrary{calc}
\usetikzlibrary{automata}
\usetikzlibrary{positioning}

\tikzstyle{label}=[shape=circle,draw,inner sep=0pt,minimum size=5mm]
\tikzstyle{tran}=[draw,->,>=stealth, rounded corners]
\tikzstyle{line} = [draw, color=black, -latex]

\usepackage{listings}
\usepackage{mathtools}
\usepackage{stmaryrd}

\usepackage{subcaption}

\usepackage[ruled,vlined,linesnumbered]{algorithm2e}

\lstdefinelanguage{prog}
{
morekeywords={if, then, else, end, for, true, false, and, or, skip, return},
sensitive = false
}

\DeclareMathAlphabet{\mathpzc}{OT1}{pzc}{m}{it}

\DeclareMathOperator{\cone}{cone}

\DeclareMathOperator{\eff}{\mathit{eff}}

\newcommand{\plen}{\text{L}}
\newcommand{\term}{\mathcal{L}}
\newcommand{\conf}{\mathit{C}}
\newcommand{\Inc}{\mathit{Inc}}
\newcommand{\Cone}{\mathit{cone}}

\newcommand{\A}{\ensuremath{\mathcal{A}}}
\newcommand{\calH}{\ensuremath{\mathcal{H}}}

\newcommand{\N}{\ensuremath{\mathbb{N}}}
\newcommand{\Z}{\ensuremath{\mathbb{Z}}}
\newcommand{\R}{\ensuremath{\mathbb{R}}}
\newcommand{\Q}{\ensuremath{\mathbb{Q}}}
\newcommand{\bigO}{\ensuremath{\mathcal{O}}}
\newcommand{\ce}[1]{\ensuremath{\left(#1 \right)}}

\newcommand{\size}[1]{|\!|#1|\!|}

\newcommand{\bx}{\ensuremath{\bold{x}}}
\newcommand{\by}{\ensuremath{\bold{y}}}
\newcommand{\bz}{\ensuremath{\bold{z}}}
\newcommand{\bu}{\ensuremath{\bold{u}}}
\newcommand{\bv}{\ensuremath{\bold{v}}}
\newcommand{\bw}{\ensuremath{\bold{w}}}
\newcommand{\bn}{\ensuremath{\bold{n}}}
\newcommand{\bc}{\ensuremath{\bold{c}}}

\usepackage{graphics}

\newcommand{\EXPSPACE}{\textsf{EXPSPACE}}

\newcommand{\Norm}{\mathit{norm}}
\newcommand{\Decomp}{\mathit{Decomp}}

\newcommand{\Rset}{\mathbb{R}}

\newcommand{\vass}{\A}
\newcommand{\update}{\bu}
\newcommand{\updates}{U}
\newcommand{\transition}{\ensuremath{t}}
\newcommand{\loc}{\ensuremath{p}}

\newcommand{\zeroVec}{\ensuremath{\vec{0}}}
\newcommand{\oneVec}{\ensuremath{\vec{1}}}
\newcommand{\flowMatrix}{\ensuremath{F}}
\newcommand{\parameter}{\ensuremath{n}}

\newcommand{\rationalLP}{\ensuremath{R}}
\newcommand{\optimal}{\ensuremath{c}}
\newcommand{\denominator}{\ensuremath{m}}
\newcommand{\counters}{\ensuremath{\bolden{\mu}}}
\newcommand{\countersAlt}{\ensuremath{\bolden{\rho}}}
\newcommand{\paath}{\ensuremath{\pi}}
\newcommand{\valueSum}{\eff}
\newcommand{\cycle}{\ensuremath{C}}
\newcommand{\multCycle}{\ensuremath{M}}
\newcommand{\length}{\ensuremath{l}}
\newcommand{\dimOp}{\ensuremath{\mathit{dim}}}

\newcommand{\val}{\bv}
\newcommand{\qrf}{f}
\newcommand{\normalVec}{\bc}
\newcommand{\qrfLP}{\ensuremath{Q}}
\newcommand{\activeQ}{\ensuremath{\bolden{b}}}
\newcommand{\rankCoeff}{\normalVec}
\newcommand{\offsets}{\weights}

\newcommand{\exampleCS}{\ensuremath{\mathit{csys}}}
\newcommand{\exampleProg}{\ensuremath{\mathit{prog}}}

\newcommand{\locs}{Q}
\newcommand{\trans}{T}
\newcommand{\bolden}[1]{\boldsymbol{#1}}
\newcommand{\weights}{\bw}
\newcommand{\precise}{\mathit{tight}}
\newcommand{\maxup}[1]{\max_{#1}}

\newcommand{\rationalLPdual}{\rationalLP_{\mathit{dual}}}

\setcopyright{none}

\begin{document}
	
\author{Tom\'{a}\v{s} Br\'{a}zdil}
	\affiliation{
	\institution{Faculty of Informatics, Masaryk University}
		\city{Brno}
		\country{Czech Republic}
	}
	\email{brazdil@fi.muni.cz}
	
\author{Krishnendu Chatterjee}
	\affiliation{
		\institution{IST Austria}
		\city{Klosterneuburg}
		\country{Austria}
	}
	\email{Krishnendu.Chatterjee@ist.ac.at}
	
\author {Anton\'{\i}n Ku\v{c}era}
\affiliation{
	\institution{Faculty of Informatics, Masaryk University}
	\city{Brno}
	\country{Czech Republic}
}
\email{kucera@fi.muni.cz}
	
\author{Petr Novotn\'{y}}
	\affiliation{
		\institution{IST Austria}
		\city{Klosterneuburg}
		\country{Austria}
	}
	\email{petr.novotny@ist.ac.at}
	
\author {Dominik Velan}
	\affiliation{
		\institution{Faculty of Informatics, Masaryk University}
		\city{Brno}
		\country{Czech Republic}
	}
	\email{xvelan1@fi.muni.cz}

\author {Florian Zuleger}
\affiliation{
	\institution{Forsyte, TU Wien}
	\city{Wien}
	\country{Austria}
}
\email{zuleger@forsyte.at}
	
\title[Asymptotic termination in VASS]{\bf Efficient Algorithms for Asymptotic Bounds on Termination Time in VASS}

\begin{abstract}
	Vector Addition Systems with States (VASS) provide a well-known and
	fundamental model for the analysis of concurrent processes,
	parameterized systems, and are also used as abstract models of programs in resource bound analysis. In this paper we study the problem of obtaining asymptotic bounds
	on the termination time of a given VASS. In particular, we focus on the practically
	important case of obtaining polynomial bounds on termination time.
	Our main contributions are as follows:
	First, we present a polynomial-time algorithm for deciding whether a given VASS has a linear asymptotic complexity. We also show
	that if the complexity of a VASS is not linear, it is at least quadratic.
	Second, we classify VASS according to quantitative properties of their cycles. We show that certain singularities in these properties are the key reason for non-polynomial  asymptotic complexity of VASS. In absence of singularities, we show that the asymptotic complexity is always polynomial and of the form $\Theta(n^k)$, for some integer $k\leq d$, where $d$ is the dimension of the VASS. We present a polynomial-time algorithm computing the optimal $k$. For general VASS, the same algorithm, which is based on a complete technique for the construction of ranking functions in VASS, produces a valid lower bound, i.e., a
	$k$ such that the termination complexity is $\Omega(n^k)$. Our results are based on new insights into the geometry of VASS dynamics, which hold the potential for further applicability to VASS analysis.
\end{abstract}

\maketitle

\renewcommand{\shortauthors}{T. Br\'azdil, K. Chatterjee, A. Ku\v{c}era, P.
Novotn\'y, D.
Velan, F. Zuleger}

\section{Introduction}
\label{sec-intro}

Vector Addition Systems with States (VASS) are a fundamental model widely used in program analysis. Intuitively, a VASS consists of a finite set of control states and transitions between the control states, and a set of $d$ counters that hold non-negative integer values, where at every transition between the control states each counter is updated by a fixed integer value. A \emph{configuration} $p\bv$ of a given VASS is determined by the current control state $p$ and the vector $\bv$ of current counter values.

One of the most basic problems studied in program analysis is  \emph{termination} that, given a program, asks whether it always terminates. For VASS, the problem whether all paths initiated in \emph{given} configuration reach a terminal configuration is \EXPSPACE-complete. Here, a terminal configuration is a configuration where the computation is ``stuck'' because all outgoing transitions would decrease some counter to a negative value. The \EXPSPACE-hardness follows from~\cite{Lipton:PN-Reachability}, and the upper bound from~\cite{Yen92:Petri-Net-logic,AH11:Yen}. Contrasting to this, the problem of \emph{structural} VASS termination, which asks whether \emph{all} configurations of a given VASS terminate, is solvable in polynomial time \cite{KS88}. This is encouraging, because structural termination guarantees termination for all instances of the parameters represented by the counter values (i.e., all inputs, all instances of a given parameterized system, etc.).

The \emph{quantitative} variant of the termination question asks whether a given program terminates in $\bigO(f(n))$ steps for every input of size~$n$, where $f \colon \N \rightarrow \N$ is some function.
A significant research effort has recently been devoted to this question in the program analysis literature:
Recent projects include SPEED~\cite{SPEED2,conf/pldi/GulwaniZ10}, COSTA~\cite{DBLP:journals/entcs/AlbertAGGPRRZ09},
RAML~\cite{journals/toplas/0002AH12},
Rank~\cite{ADFG10:lexicographic-ranking-flowcharts},
Loopus~\cite{SZV14,journals/jar/SinnZV17},
AProVE~\cite{journals/jar/GieslABEFFHOPSS17},
CoFloCo~\cite{conf/aplas/Flores-MontoyaH14},
C4B~\cite{conf/cav/Carbonneaux0RS17}.
The cited projects target general-purpose programming languages with the goal of designing \emph{sound (but incomplete)} analyses that work well in practice.
The question whether \emph{sound and complete} techniques can be developed for restricted classes of programs (such as VASS), however, has received considerably less attention.

\textbf{Our contribution.} In this work, we study the quantitative variant of structural VASS termination. The \emph{termination complexity} of a given VASS is a function $\term \colon \N \rightarrow \N \cup \{\infty\}$ such that $\term(n)$ is the length of the longest computation initiated in a configuration $p\bv$ where all components of $\bv$ are bounded by $n$.  We concentrate on \emph{polynomial} and particularly on \emph{linear} asymptotic bounds for termination complexity, which seem most relevant for practical applications. Our main results can be summarized as follows:

\emph{Linear bounds.} We show that the problem whether $\term \in \Theta(n)$ is decidable in polynomial time. Our proof reveals that if the termination complexity is \emph{not} linear, then it is at least quadratic (or the VASS is non-terminating). Hence, there is no VASS with asymptotic termination complexity ``between'' $\Theta(n)$ and $\Theta(n^2)$. In addition, for strongly connected linear VASS, we compute a constant $c \in \mathbb{Q}$ (in polynomial time)  such that $\term(n) = cn$ for $n \rightarrow \infty$.
Further, a linear VASS always has a ranking function that witnesses the linear termination complexity; this ranking function is also computable in polynomial time.

\emph{Polynomial bounds.} We show that the termination complexity of a given VASS is highly influenced by the properties of \emph{normals of quasi-ranking functions}, see Section~\ref{sec-poly}. We start with strongly connected VASS, and classify them into the following three types:
\begin{compactitem}
	\item[(A)] Non-terminating VASS.
	\item[(B)] Positive normal VASS: Terminating VASS for which there
	exists a quasi-ranking function such that each component of its normal is positive.
	\item[(C)] Singular normal VASS: Terminating VASS for which there
	exists a quasi-ranking function such that each component of its normal is non-negative
	and~(B) does not hold.
\end{compactitem}
This classification is efficient, i.e., we can decide in polynomial time to which class a given VASS belongs. We show that each type~(B) VASS of dimension $d$ has termination complexity in $\Theta(n^k)$, where $1\leq k \leq d$, and we show that the $k$ is computable in polynomial time. Termination complexity of a type~(C) VASS is not necessarily polynomial, and hence singularities in the normal are the key reason for high asymptotic bounds in VASS.  For a given type~(C) VASS, we show how to compute a valid lower bound, i.e., a $k$ such that the termination complexity is $\Omega(n^k)$ (in general, this bound does not have to be tight). Our tight analysis for type~(B) VASS extends to general (not necessarily strongly connected) VASS where each SCC determines a type~(B) VASS.

\emph{Ranking Functions and Completeness.}
Algorithmically the result on polynomial bounds is established by a recursive procedure:
the procedure computes quasi-ranking functions which establish that certain transitions can only be taken a linear number of times; these transitions are then removed and the algorithm recurses on the remaining strongly-connected components.
We show that if there is no quasi-ranking function, then the VASS does not terminate, i.e., our ranking function construction is complete.
To the best of our knowledge, this is the first \emph{completeness} result for the construction of ranking functions for VASS.

Technically, our results are based on new insights into the geometry of VASS dynamics, some of which are perhaps interesting on their own and can enrich the standard toolbox of techniques applicable to VASS analysis.

\begin{figure}[t]
\begin{tabular}{l|l}
\begin{minipage}{4cm}
 \begin{alltt}
  void main(uint n) \{
    uint i = n, j = n;
\(l\sb{1}\!:\)  while (i > 0) \{
      i--;
      j++;
\(l\sb{2}\!:\)    while (j > 0 && *)
          j--;
  \} \}
\end{alltt}
\end{minipage}
&
\begin{minipage}{2,5cm}
\begin{tikzpicture}[scale=0.4, node distance = 1cm, auto]
    \node (t1) [xshift=2cm] {$\loc_1$};
    \node (t2) [below of=t1,node distance=1cm] {$\loc_2$};
    \path(t1) edge [line,bend left] node [right]
    {(-1,1)}(t2)
    (t2) edge [loop below,  every loop/.append style={looseness= 20}] node [right, yshift=0.3cm,xshift=0.1cm]
    {(0,-1)}(t2)
    (t2) edge [line,bend left] node [left]
    {(0,0)} (t1)
    ;
   \end{tikzpicture}
   \end{minipage}
\vspace{-0.3cm}
\end{tabular}
\caption{(a) a program, (b) VASS $\vass_\exampleProg$}
\vspace{-0.2cm}
\label{fig:intro-prog}
\end{figure}

\begin{figure}[t]
\begin{tabular}{l|l}
\begin{minipage}{4cm}
\begin{tikzpicture}[scale=0.4, node distance = 1cm, auto]
    \node (t1) [xshift=2cm] {$i$};
    \node (t2) [below of=t1,node distance=1.5cm] {$j$};
    \node (t3) [right of=t1,node distance=1cm] {$k$};
    \path(t1) edge [line] node [above]
    {d==ff, d:=tt}(t3)
    (t1) edge [line,bend left] node [right]
    {
    {d:=ff}
    }(t2)
    (t2) edge [line,bend left] node [left]
    {d==tt} (t1)
    ;
   \end{tikzpicture}
   \end{minipage}
&
\begin{minipage}{2,5cm}
\begin{tikzpicture}[scale=0.4, node distance = 1cm, auto]
    \node (t1) [xshift=2cm] {$\loc_{d=\mathtt{tt}}$};
    \node (t2) [below of=t1] {$\loc_{d=\mathtt{ff}}$};
    \path(t1) edge [line,bend left] node [right]
    {(-1,1,0)}(t2)
    (t2) edge [loop below,  every loop/.append style={looseness= 10}] node [right, yshift=0.3cm,xshift=0.3cm]
    {(-1,1,0)}(t2)
    (t2) edge [line,bend left] node [left]
    {(-1,0,1)} (t1)
    (t1) edge [loop above,  every loop/.append style={looseness= 10}] node [right, yshift=-0.3cm,xshift=0.3cm]
    {(1,-1,0)}(t1)
    ;
   \end{tikzpicture}
\end{minipage}
\vspace{-0.3cm}
\end{tabular}
\caption{(a) a process template, (b) VASS $\vass_\exampleCS$}
\vspace{-0.2cm}
\label{fig:intro-parameterized}
\end{figure}

\paragraph{Motivation and Illustration of our Results.}
In previous work we have described automated techniques for the complexity analyses of imperative programs, which use VASS (and extensions) as backend~\cite{SZV14,journals/jar/SinnZV17}.
For example, our techniques allow to abstract the program given in Fig.~\ref{fig:intro-prog} (a) to the VASS $\vass_\exampleProg$ in Fig.~\ref{fig:intro-prog} (b).
$\vass_\exampleProg$ has two locations $\loc_1$ and $\loc_2$, which correspond to the loop headers of the program.
$\vass_\exampleProg$ has dimension two in order to represent the variables $i$ and $j$.
The transitions of $\vass_\exampleProg$ correspond to the variable increments/decrements.
In contrast to our previous approaches~\cite{SZV14,journals/jar/SinnZV17},
the analysis in this paper is guaranteed to compute tight bounds:
we obtain the precise linear termination complexity $\term(n) = 4n$ for $\vass_\exampleProg$ and can construct a linear ranking function, e.g., $\qrf(\loc,(i,j)) = 3i + j + \offsets(\loc)$, where $\offsets(\loc_1) = 0$ and $\offsets(\loc_1) = 1$
(our construction is not guaranteed to return this ranking function, but it 
will always find a linear ranking function).

We illustrate VASSs as models of concurrent systems:
Fig.~\ref{fig:intro-parameterized} (a) states a process template.
A concurrent system consists of $n$ copies of this process template.
The processes communicate via the Boolean variable $d$.
The concurrent system is equivalently represented by the VASS $\vass_\exampleCS$ in Fig.~\ref{fig:intro-parameterized}  (b).
$\vass_\exampleCS$ has two locations $\loc_{d=\mathtt{tt}}$ and $\loc_{d=\mathtt{ff}}$, which represent the global state.
$\vass_\exampleCS$ has dimension three in order to represent the number of processes in the local states $i$, $j$ and $k$.
The transitions of $\vass_\exampleCS$ reflect the transitions of the process template, e.g., transition $(-1,1,0)$ means that one process moves from state $i$ to $j$.
We are interested in the \emph{parameterized verification problem}, i.e., to study the termination of the concurrent system for all system sizes $n$.
Our results in this paper establish $\term(n) \in \Theta(n^2)$, i.e., after quadratically many steps of the concurrent system there is no more process that can take another step.

\textbf{Related Work.}
\noindent{\em Results on VASS.}
The model of VASS~\cite{KM69} or equivalently Petri nets are a fundamental model
for parallel programs~\cite{EN94,KM69} as well as parameterized systems~\cite{Bloem16,conf/icalp/AminofRZS15,conf/lpar/AminofRZ15}.
The termination problems (counter-termination, control-state termination) as well
as the related problems of boundedness and coverability have been a rich source
of theoretical problems
that have been widely
studied~\cite{Lipton:PN-Reachability,Rackoff:Covering-TCS,Esparza:PN,ELMMN14:SMT-coverability,BG11:Vass}.
The complexity of the termination problem with fixed initial configuration is
\EXPSPACE-complete~\cite{Lipton:PN-Reachability,Yen92:Petri-Net-logic,AH11:Yen}.
Besides the termination problem, the more general reachability problem where given a VASS,
an initial and a final configuration, whether there exists a path between them has also been
studied~\cite{Mayr:PN-reachability,Kosaraju82:VASS-reach-dec,Leroux:VASS-POPL}.
The reachability problem is
decidable~\cite{Mayr:PN-reachability,Kosaraju82:VASS-reach-dec,Leroux:VASS-POPL},
and
\EXPSPACE-hard~\cite{Lipton:PN-Reachability},
and the current best-known upper bound is cubic
Ackermannian~\cite{LS15:VASS-reach-complexity}, a complexity class belonging to
the third level of a fast-growing complexity hierarchy introduced
in~\cite{Schmitz16:hyperackermannian-complexity-hierarchy}. Functions (non)computable by VASS are studied in \cite{DBLP:conf/rp/LerouxS14}.
Our algorithm for computing polynomial bounds can be seen as the dual (in the sense of linear programming) of the algorithm of~\cite{KS88};
this connection is the basis for the completeness of our ranking function construction (we further comment on the connection to~\cite{KS88} in Section~\ref{sec-poly}).

\smallskip\noindent{\em Ranking functions and extensions.}
Ranking functions for intraprocedural analysis have been studied widely in the literature.
We restrict ourselves here to approaches which present complete methods for the construction of linear/polynomial ranking functions~\cite{DBLP:conf/vmcai/PodelskiR04,ADFG10:lexicographic-ranking-flowcharts,DBLP:journals/fcsc/YangZZX10};
in contrast to this paper these approaches target general programs and do not show that the non-existence of a linear/polynomial ranking function implies the non-termination of the program.

The problem of existence of infinite computations in VASS has been studied in
the literature.
Polynomial-time algorithms have been presented in~\cite{CDHR10,VCDHRR15} using results of~\cite{KS88}.
In the more general context of games played on VASS, even deciding the existence of
infinite computation is coNP-complete~\cite{CDHR10,VCDHRR15}, and various algorithmic approaches
based on hyperplane-separation technique have been studied~\cite{CV13,JLS15,CJLS17}.

\section{Preliminaries}
\label{sec-prelim}

We use $\N$, $\Z$, $\Q$, and $\R$ to denote the sets of non-negative integers,
integers,
rational numbers, and real numbers. The subsets of all \emph{positive} elements
of $\N$, $\Q$, and $\R$ are denoted by $\N^+$, $\Q^+$, and $\R^+$.  Further, we
use $\N_\infty$ to denote the set $\N \cup \{\infty\}$ where $\infty$ is
treated according to the standard conventions.
The cardinality of a given set $M$ is denoted by $|M|$. When no confusion arises, we also use $|c|$ to denote the absolute value of a given $c \in \R$.

Given a function $f \colon \N \rightarrow \N$, we use  $\bigO(f(n))$ and $\Omega(f(n))$ to denote the sets of all $g \colon \N \rightarrow \N$ such that $g(n) \leq a \cdot f(n)$ and $g(n) \geq b \cdot f(n)$ for all sufficiently large $n \in \N$, where $a,b \in \R^+$ are some constants. If $h(n) \in \bigO(f(n))$ and $h(n) \in \Omega(f(n))$, we write $h(n) \in \Theta(f(n))$.

Let $A,B$ be arbitrary index sets. Elements of $\R^A$ are denoted by bold
letters such as
$\bu,\bv,\bz,\ldots$.
The component of $\bv$ of index $i\in A$ is denoted by $\bv(i)$. For a matrix
$A \in \Rset^{A\times
B}$ we denote by $A(a,b)$ the element in row of index $a\in A$ and column
of index by $b\in
B$, and by $A^{\top}$ the transpose of $A$. If the index set is of the form $A=\{1,2,\dots,d\}$ for some positive
integer $d$, we write $\Rset^d$ instead of $\Rset^A$, i.e., for $\bv\in \Rset^d$
we have $\bv = (\bv(1),\dots,\bv(d))$.
For every $n \in \N$, we use $\vec{n}$ to denote the constant vector where all
components are equal to~$n$.
The scalar product of $\bv,\bu \in \R^d$ is denoted by $\bv \cdot \bu$, i.e.,
$\bv\cdot \bu = \sum_{i=1}^d \bv(i)\cdot\bu(i)$. The other
standard operations and relations on $\R$ such as
$+$, $\leq$, or $<$ are extended to $\R^d$ in the component-wise way. In
particular, $\bv$ is \emph{positive} if $\bv > \vec{0}$, i.e., all components of
$\bv$ are positive. The norm of $\bv$ is defined by $\Norm(\bv) = \sqrt{\bv(1)^2 + \cdots+\bv(d)^2}$.

\paragraph{Half-spaces and Cones.}
An \emph{open half-space} of $\R^d$ determined by a normal vector $\bn \in
\R^d$, where $\bn \neq \vec{0}$, is the set $\calH_{\bn}$ of all $\bx \in \R^d$ such that $\bx \cdot \bn < 0$. A \emph{closed half-space} $\hat{\calH}_{\bn}$ is defined in the same way
but the above inequality is non-strict. Given a finite set of vectors $U
\subseteq \R^d$, we use $\Cone(U)$ to denote the set of all vectors of the form
$\sum_{\bu \in U} c_{\bu} \bu$, where $c_{\bu}$ is a non-negative real constant
for every $\bu \in U$.

\subsection{Syntax and semantics of VASS}
\label{sec-VASS-def}

In this subsection we present a syntax of VASS, represented as finite state
graphs with transitions labelled by vectors of counter changes.

\begin{definition}
\label{def-VASS}
Let $d \in \N$. A \emph{$d$-dimensional vector addition system with states
(VASS)} is a pair $\A = \ce{Q,T}$, where $Q \neq \emptyset$ is a finite set of
\emph{states} and  $T \subseteq Q \times \Z^d \times Q$ is a finite set of
\emph{transitions} such that for every $q \in Q$ there exists $p \in Q$
 and $\bu \in \Z^d$ such that $(q,\bu,p) \in T$.
\end{definition}

We denote by $\maxup{\A}$ the number $\max_{(p,\bu,q)\in T,1\leq i \leq d}|\bu(i)|$. The encoding size of $\A$ is denoted by $\size{\A}$ (the integers representing counter updates are written in binary).

In our disucssion it is often beneficial to express constraints on transitions using matrix notation.
We define the \emph{update} matrix $\updates \in \mathbb{Z}^{d \times T}$ by
setting $\updates(i, \transition) = \update(i)$ for all $1\leq i\leq d$ and all transitions $\transition =
(\loc,\update,\loc') \in \trans$.
We also define the \emph{oriented incidence matrix} $\flowMatrix \in
\mathbb{Z}^{\locs \times \trans}$ by setting
$\flowMatrix(\loc,\transition) = 1$ resp. $\flowMatrix(\loc,\transition) = -1$,
if $\transition = (\loc,\update,\loc')$ resp. $\transition =
(\loc',\update,\loc)$ and $\loc' \neq \loc$, and $\flowMatrix(\loc,\transition)
= 0$, otherwise.
We note that every column of $\flowMatrix$, corresponding to a transition $\transition$, either contains exactly one $-1$ entry and exactly one $1$ entry (in case the source and target of transition~$\transition$ are different) or only $0$ entries (in case the source and target of transition~$\transition$ are the same).

\begin{example}
VASS $\vass_\exampleProg$ from Fig.~\ref{fig:intro-prog} (b) has two states $\loc_1,\loc_2$ and three
transitions
$t_1 = (q_1,(-1,1),q_2)$, $t_2 = (q_2,(0,0),q_1)$, $t_3 = (q_2,(0,-1),q_2)$. The matrices $F$ and $U$ look as follows:
\[
F=\begin{pmatrix}
1 & -1 & 0\\
-1 & 1 & 0
\end{pmatrix}
\]
Here the rows correspond to the states $q_1, q_2$ and columns to transitions $t_1,t_2,t_3$.
\[
U=\begin{pmatrix}
-1 & 0 & 0\\
1 & 0 & -1
\end{pmatrix}
\]
Hence, the columns are the update vectors of transitions $t_1,t_2,t_3$.

\end{example}

\paragraph{Paths and cycles.}
A \emph{finite path} in $\A$ of length~$n$ is a finite sequence $\pi$ of the form
$p_0,\bu_1,p_1,\bu_2,p_2,\ldots,\bu_n,p_n$ where $n \geq 1$ and
$(p_i,\bu_{i+1},p_{i+1}) \in T$ for all $0 \leq i < n$. If $p_0 = p_n$, then
$\pi$ is a \emph{cycle}. A~cycle is \emph{simple} if all $p_1,\ldots, p_{n-1}$ are pairwise different. The \emph{effect} of $\pi$, denoted by $\eff(\pi)$, is the sum $\bu_1 + \cdots + \bu_n$. Given a set of paths $P$, we denote by $\eff(P)$ the sum of effects of all paths in $P$.
Let
\(
   \Inc = \{\eff(\pi) \mid \pi \text{ is a simple cycle of } \A\}\,.
\)
The elements of $\Inc$ are called \emph{increments}.

Given two finite paths $\alpha = p_0,\bu_1,\ldots,p_n$ and $\beta = q_0,\bv_1,\ldots,q_m$ such that $p_n = q_0$, we use $\alpha \odot \beta$ to denote the finite path $p_0,\bu_1,\ldots,p_n,\bv_1,\ldots,q_m$.
A \emph{multi-cycle} in $\A$ is a multiset of simple cycles.
The length of a multi-cycle is the sum of lengths of all its cycles.

Let $\pi$ be a finite path in $\A$. A \emph{decomposition of $\pi$ into simple cycles}, denoted by $\Decomp(\pi)$, is a multi-cycle, i.e., a multiset of simple cycles, defined recursively as follows:
\begin{compactitem}
	\item If $\pi$ does not contain any simple cycle, then $\Decomp(\pi)$ is an empty multiset.
	\item If $\pi = \alpha \odot \gamma \odot \beta$ where $\gamma$ is the first simple cycle occurring in $\pi$, then $\Decomp(\pi) = \{\gamma\}\cup\Decomp(\alpha \odot \beta)$.
\end{compactitem}
Observe that if $\Decomp(\pi)$ is empty, then the length of $\pi$ is at most $|Q|-1$. Since the length of every simple cycle is bounded by $|Q|$, the length of $\pi$ is asymptotically the \emph{same} as the number of elements in $\Decomp(\pi)$, assuming a fixed VASS~$\A$.
Considering $\pi'$ to be the remainder of $\pi$ after all simple cycles of $\Decomp(\pi)$ removed by the above procedure, we obtain
$\valueSum(\pi)=
\valueSum(\pi')+\valueSum(\Decomp(\pi))$.

Let $\A = (Q,T)$ be a VASS. A \emph{sub-VASS} of $\A$ is a VASS $\A' = (Q',T')$ such that $Q' \subseteq Q$ and $T' \subseteq T$. VASS $\A$ is \emph{strongly connected} if for every $p,q \in Q$ there is a finite path from $p$ to $q$.

A \emph{strongly connected component (SCC)} of $\A$ is a maximal strongly connected sub-VASS of $\A$.

\paragraph{Configurations and computation.}
A \emph{configuration} of $\A$ is a pair $p\bv$, where $p \in Q$ and $\bv \in
\N^d$. The set of all configurations of $\A$ is denoted by $\conf(\A)$. The
\emph{size} of  $p\bv \in \conf(\A)$ is defined as $\size{p\bv} = \size{\bv} =
\max \{\bv(i)
\mid 1\leq i \leq d\}$. Given $n\in \mathbb{N}$, we say that $p\bv$ is {\em $n$-bounded} if $\size{p\bv}\leq n$.

A \emph{computation} initiated in $p_0\bv_0$ is a finite sequence $p_0\bv_0,\ldots,p_n \bv_n$ of configurations such that there exists a path $p_0,\bu_1,p_1,\bu_2,p_2,\ldots,\bu_n,p_n$  where $\bv_i = \bv_0 + \bu_1 + \cdots + \bu_i$ for all $0 \leq i \leq n$. The length of a given computation is the length of its (unique) corresponding path.

\subsection{Termination Complexity of VASS}
\label{sec-VASS-term}

\begin{definition}
\label{def-VASS-compl}
Let $\A = \ce{Q,T}$ be a $d$-dimensional VASS. For every configuration $p\bv$ of $\A$, let $\plen(p\bv)$ be the least $\ell \in \N_\infty$ such that the length of every finite computation initiated in $p\bv$ is bounded by $\ell$. The \emph{termination complexity} of $\A$ is a function $\term \colon \N \rightarrow \N \cup \{\infty\}$ defined by
\(
   \term(n) = \sup \left\{\plen(p\bv) \mid p\bv \in \conf(\A) \mbox{ where } \size{p\bv} = n \right\} .
\)
If $\term(n) = \infty$ for some $n \in \N$, we say that $\A$ is \emph{non-terminating}, otherwise it is \emph{terminating}.
\end{definition}	
Observe that if $\A$ is non-terminating, then $\term(n) = \infty$ for all sufficiently large $n \in \N$. Further, if $\A$ is terminating, then $\term(n) \in \Omega(n)$. In particular, if $\term(n) \in \bigO(n)$, we also have $\term(n) \in \Theta(n)$.

\section{Linear Termination Time}
\label{sec-linear}

In this section, we give a complete and effective characterization of all VASS
with \emph{linear} termination complexity. 
Let us consider a VASS $\A = \ce{Q,T}$. We assume that $\A$ is strongly connected unless explicitly stated otherwise.

Consider an integer solution $\counters \in \mathbb{Z}^{\trans}$ to the
constraints $\counters \ge \zeroVec$ and $\flowMatrix \counters = \zeroVec$ (here $\flowMatrix$ is the oriented incidence matrix of $\A$).
Note that $\counters$ induces a \emph{multi-cycle}
$\multCycle$ of $\A$. Indeed, if $\counters(t)>0$, then there is a transition $t'$ with $\counters(t')>0$ such that the source state of $t'$ is equal to the target state of $t$.
Hence one may trace a path over states with positive value in $\counters$ that eventually leads to a simple cycle. Subtracting one from $\counters(t)$ for all $t$ on the simple cycle we obtain  $\counters'$ still satisfying the above constraints. Repeating this process we eventually end up with a~zero vector and the desired multi-cycle~$\multCycle$. 

Note that $\oneVec^{\top}\counters$ is equal to the number of transitions traced along the multi-cycle. So, roughly speaking, it suffices to add a constraint $U\counters \geq -n \oneVec$ (here $U$ is the update matrix) to characterize multicycles that, when appropriately executed in an $n$-bounded configuration, produce a zero-avoiding computation. However, there are several issues in such a formulation, namely dependency of the constraints on the parameter $n$ and demand for an integer solution.

So we transform the constraints into the following relaxed optimization problem to completely characterize the linear computational complexity:

\begin{center}
\begin{tabular}{|c|}
\hline

{\begin{minipage}[c]{0.5\linewidth}
\vspace{0.2cm}

rational LP ($\rationalLP$):

\vspace{0.2cm}
$\max \oneVec^\top\countersAlt$  with
\begin{align*}
  \countersAlt & \ge \zeroVec\\
  \updates \countersAlt & \ge -\oneVec\\
  \flowMatrix \countersAlt & = \zeroVec
\end{align*}
\end{minipage}}\\
\hline
\end{tabular}
\end{center}

\begin{theorem}
\label{thm:precise-complexity}
  Let $\vass$ be a strongly connected VASS.
  We consider LP ($\rationalLP$) over $\mathbb{Q}$.
	\begin{compactenum}  
\item[(A)]   If LP ($\rationalLP$) has a solution $\countersAlt$ with $\oneVec^\top\countersAlt=\optimal \in \mathbb{Q}$, then
  $\optimal \parameter$ is the precise asymptotic computational complexity of
  $\vass$, i.e., $\term(\parameter)$ converges to $\optimal \parameter$ for
  $\parameter \rightarrow \infty$.
  \item[(B)] If ($\rationalLP$) is unbounded, then the computational complexity of $\vass$
  is at least quadratic.
  \end{compactenum}
\end{theorem}

{\bf Intuition:} Let $\countersAlt$ be a rational solution of ($\rationalLP$) with $\oneVec^\top\countersAlt=\optimal$ and consider a non-negative integer $n\in \mathbb{N}$. 
Let $\counters \in \mathbb{Z}^{\trans}$ satisfy $\counters = n m \countersAlt$ where $m$ is the least common multiple of the denominators of $\countersAlt$. Since $\counters = nm\countersAlt \geq \zeroVec$ and $\flowMatrix \counters = \flowMatrix nm\countersAlt = \zeroVec$, the vector $\counters$ specifies a multi-cycle of length $\oneVec^\top \counters = c n m$. Moreover, $\counters$ satisfies
$\updates \counters = \updates n m \countersAlt \geq -n m \cdot \oneVec$ which means that executing all transitions of the multi-cycle cannot decrease the counters by more than $m n$. By executing cycles 
of the multi-cycle in a carefuly arranged order initiated in a $n$-bounded configuration, we obtain a zero-avoiding computation 
 whose length is, roughly, $c n$.

On the other hand, if the program ($\rationalLP$) is unbounded, we show that then there is a solution $\countersAlt$ satisfying $U\countersAlt\geq\zeroVec$. From this we obtain multi-cycles of arbitrary length whose overall effect is non-negative. Note that this does not mean that the VASS is non-terminating since the cycles need to be connected into a single computation. However, we show that they always can be connected into a computation of at least quadratic length.

{\bf Proof of Theorem~\ref{thm:precise-complexity} (A).}
Assume ($\rationalLP$) is bounded.
Let $\countersAlt \in \mathbb{Q}^{T}$ be an optimal solution.
We set $\optimal = \oneVec^\top\countersAlt$.
We first show the upper bound.
We fix some $\parameter$.
We consider the longest computation starting from some $\parameter$-bounded configuration.
Let $\paath$ be the path associated to this computation.
Because we are interested only in asymptotic behaviour, we can assume $\paath$ is a cycle.
Let $\counters_\paath(\transition)$ denote the number of occurrences of transition $\transition$ on $\paath$.
We note that $\updates \counters_\paath = \valueSum(\paath) \ge -\parameter \cdot \oneVec$ because the starting configuration of the considered worst-case computation is $\parameter$-bounded.
Because $\paath$ is a cycle, we have $\flowMatrix \counters_\paath = \zeroVec$.
Hence, $\frac{1}{\parameter} \cdot \counters_\paath$ is a feasible point of LP ($\rationalLP$) and we get $\oneVec^\top\frac{1}{\parameter} \cdot \counters_\paath \le \optimal$.
Thus, $\oneVec^\top\counters_\paath \le \optimal\parameter$.
Because this holds for all $\parameter$, we can conclude $\term(\parameter) \le \optimal\parameter$.

We show the lower bound.
We fix some $\parameter$.
Let $m$ be the least common multiple of the denominators of $\countersAlt$.
We set $\counters = \denominator \cdot \countersAlt \in \mathbb{Z}^{T}$.
We have $\counters \ge \zeroVec$, $\updates \counters \ge -\denominator \cdot \oneVec$, $\flowMatrix \counters = \zeroVec$ and $\oneVec^\top\counters = \optimal \denominator$.
We consider the multi-cycle $\multCycle$ associated to $\counters$.
Let $\cycle$ be some cycle of $\vass$ which visits each state at least once.
Let $\length$ be the length of $\cycle$.
Because $\cycle$ visits every state at least once we can combine $\cycle$ and $\sqrt{\parameter}$ copies of multi-cycle $\multCycle$ into a single cycle $\cycle'$.
Let $\length'$ be the length of $\cycle'$.
We have $\length' = \length + \sqrt{\parameter}\oneVec^\top\counters = \length + \sqrt{\parameter} \optimal \denominator$.
Let $\loc$ be the start and end state of $\cycle'$.
We set $\parameter' = \frac{\parameter - (\length + \sqrt{\parameter} 
	\optimal \denominator) \cdot\maxup{A} }{ \maxup{A} \cdot\length + 
	\denominator 
	\sqrt{\parameter}}$ (rounded down if needed).
Let $\val_0 = \parameter \cdot \oneVec$.
We show that starting from configuration $\loc\val_0$ we can $\parameter'$ times
execute the cycle $\cycle'$.
This is sufficient to establish $\term(\parameter) = \optimal \parameter$ because of $\frac{\parameter' \length'}{\optimal \parameter} \rightarrow 1$ for $\parameter \rightarrow \infty$.

We consider the configurations $\loc\val_i$ after $0 \le i < \parameter'$ executions of $\cycle'$.
We show by induction on $i$ that $\cycle'$ can be executed one more time.
We have $\valueSum(\cycle') = \valueSum(\cycle) + 
\sqrt{\parameter}\valueSum(\multCycle) = \valueSum(\cycle) + 
\sqrt{\parameter} \updates \counters \ge - (\maxup{A}\cdot \length + 
\denominator 
\sqrt{\parameter}) \cdot \oneVec$.
Hence, we have $\val_i \ge \parameter \cdot \oneVec - i (\maxup{A} 
\cdot\length + 
\denominator \sqrt{\parameter})\cdot \oneVec$.
We have to show that we can execute $\cycle'$ one more time.
In every step of $\cycle'$ we decrease each vector component by at most 
$\maxup{A}$.
Hence, we need to show $\val_i \ge \length'\cdot \maxup{A} \cdot \oneVec$.
Indeed, we have $\val_i \ge \parameter \cdot \oneVec - i (\maxup{A} \cdot
\length + \denominator \sqrt{\parameter}) \cdot \oneVec \ge (\length + 
\sqrt{\parameter} \optimal \denominator) \cdot \maxup{A} \cdot \oneVec$.

{\bf Proof of Theorem~\ref{thm:precise-complexity} (B).}
Assume ($\rationalLP$) is unbounded.
We will show that there is no open half-space $\calH_{\bn}$ of $\R^d$ such that $\bn > \zeroVec$ and  $\Inc \subseteq \calH_{\bn}$.
As we show later, this implies that the computational complexity of $\vass$ is at least quadratic.
From the theory of linear programming we know that there is a direction in which the polyhedron given by $\countersAlt \ge \zeroVec$, $\updates \countersAlt \ge -\oneVec$ and $\flowMatrix \countersAlt = \zeroVec$ is unbounded and which increases the objective function $\oneVec^\top\countersAlt$.
Hence, there is a $\countersAlt \ge \zeroVec$ with $\updates \countersAlt \ge \zeroVec$ and $\flowMatrix \countersAlt = \zeroVec$ and $\countersAlt(\transition) \ge 1$ for some $\transition \in T$.
We consider the multi-cycle $\multCycle$ extracted from the integer vector $\counters = m\countersAlt$ where $m$ is the common multiple of denominators in $\countersAlt$.
Assume now for the sake of contradiction that there is an open half-space $\calH_{\bn}$ of $\R^d$ such that $\bn > \zeroVec$ and $\Inc \subseteq \calH_{\bn}$.
Let $\cycle_1, \ldots,\cycle_k$ be all simple cycles occuring in $\multCycle$.
Because of $\Inc \subseteq \calH_{\bn}$ we have $\bn^\top\cdot \valueSum(\cycle_i) < 0$ for all $i$, and hence
\[
0 > \sum_i \bn^\top \valueSum(\cycle_i) = \bn^\top
\valueSum(\multCycle) = \bn^\top \updates \counters =
\bn^\top \updates m \countersAlt = m (\bn^\top \updates\counters)
\]
which implies $\bn^\top\updates\countersAlt<0$. 
On the other hand, we get $\bn^\top\cdot \updates \countersAlt \ge 0$ from $\bn > \zeroVec$ and $\updates \countersAlt \ge \zeroVec$.
A contradiction.

Now suppose there is no open half-space $\calH_{\bn}$ of $\R^d$ such that
$\bn > \vec{0}$ and  $\Inc \subseteq \calH_{\bn}$. We show that $\term(n) \in
\Omega(n^2)$, i.e., there exist $p \in Q$ and a constant $a \in \R^+$ such
that for all configurations $p \vec{n}$, where $n \in \N$ is sufficiently
large, there is a computation initiated in $p\vec{n}$ whose
length is at least $a \cdot n^2$.

The crucial point is that now there are
$\bv_1,\ldots,\bv_k \in \Inc$ and $b_1,\ldots,b_k \in \N^+$ such that $k \geq 1$ and

\begin{equation}\label{eq:growingCycles}
\sum_{i=1}^k b_i\bv_i \geq \vec{0}.
\end{equation}

The above is a direct consequence of the following purely geometric lemma (proved in appendix) with $X = \Inc_\A$.

\begin{lemma}
	\label{lem-no-halfspace}
	Let $X \subset \Z^d$ be a finite set. If there is no open\linebreak half-space 
	$\calH_{\bn}$ of\, $\R^d$ such that $\bn > \vec{0}$ and  $X \subseteq 
	\calH_{\bn}$, then there exist ${\bv_1,\ldots,\bv_k \in X}$ and 
	$b_1,\ldots,b_k \in \N^+$ such that $k \geq 1$ and $\sum_{i=1}^k b_i\bv_i 
	\geq \vec{0}$.
\end{lemma}

As the individual simple cycles with effects
$\bv_1,\ldots,\bv_k$ may proceed through disjoint sets of states,
they \emph{cannot} be trivially concatenated into one large
cycle with non-negative effect. Instead, we fix a control state $p \in Q$ and
a cycle $\pi$ initiated in $p$ visiting \emph{all} states of~$Q$. Further, for every $1 \leq i \leq k$
we fix a simple cycle $\gamma_i$ such that $\eff(\gamma_i) = \bv_i$. For every
$t \in \N$, let $\pi_t$ be a cycle obtained from $\pi$ by inserting precisely
$t\cdot b_i$ copies of every $\gamma_i$, where $1 \leq i \leq k$. Observe
that the inequality~(\ref{eq:growingCycles}) implies
\begin{equation}\label{eq:iteratingCycles}
\eff(\pi_t) = \eff(\pi) + t\cdot \sum_{i=1}^k b_i\bv_i \geq \eff(\pi)  \quad \text{ for every } t \in \N.
\end{equation}
For every configuration
$p\bu$, let $t(\bu)$ be the largest $t \in \N$ such that $\pi_t$ is
executable in $p\bu$. If such a $t(\bu)$ does not exist, i.e. $\pi_t$ is executable in $p\bu$ for all $t\in
\N$, then $\A$ is non-terminating (since, e.g. $\bv_1$ must be non-negative
in such a case), and the proof is finished. Hence, we can assume that
$t(\bu)$ is well-defined for each $\bu$. Since the
cycles $\pi$ and
$\gamma_1,\ldots,\gamma_k$ have fixed
effects, there is $b \in \R^+$ such that for all configurations $p\bu$ where
all components of $\bu$ (and thus also $\size{p\bu}$) are above some
sufficiently large threshold $\xi$
we have that $t(\bu) \geq b \cdot
\size{p\bu}$, i.e. $t(\bu)$ grows asymptotically at least linearly with the
minimal component of $\bu$. Now, for every $n \in \N$, consider a computation
$\alpha(n)$ initiated in $p \vec{n}$ defined inductively as follows:
Initially, $\alpha(n)$ consists just of  $p\bu_0 = p \vec{n}$; if the
prefix of $\alpha(n)$ constructed so far ends in a configuration $p\bu_i$
such that $t(\bu_i) \geq 1$ and $\bu_i \geq \vec{\xi}$ (an
event we call a \emph{successful hit}), then
the prefix is prolonged by executing the
cycle $\pi_{t(\bu_i)}$ (otherwise, the construction of $\alpha(n)$ stops).
Thus, $\alpha(n)$ is obtained from $p\vec{n}$ by applying the inductive rule
$I(n)$ times, where $I(n) \in \N_{\infty}$ is the number of successful hits
before the construction of $\alpha(n)$ stops.
Denote by $p\bu_i$ the configuration visited by $\alpha(n)$ at $i$-th
successful hit.
Now the inequality~(\ref{eq:iteratingCycles}) implies that
$\bu_i \geq \vec{n} + i \cdot \eff(\pi)$, so there exists a constant $e$ such
that $\size{p\bu_i}\geq n - i \cdot e$. In
particular the decrease of all components
of $\bu_i$ is at most linear in~$i$. This means that $I(n) \geq c \cdot n$
for all sufficiently large $n \in \N$, where $c \in \R^+$ is a suitable
constant. But at the same time, upon each successful hit we have $\bu_i \geq
\vec{\xi}$, so the length of the segment beginning with the $i$-th successful hit and
ending with the $(i+1)$-th hit or with the last configuration of $\alpha(n)$
is at least $b\cdot \size{p\bu_i}\geq b\cdot( n-i\cdot e)$.
Hence, the length of $\alpha(n)$ is at least $\sum_{i=1}^{c\cdot n}b\cdot
(n-i\cdot e)$, i.e. quadratic.
\qed

Finally, let us consider an arbitrary VASS $\A$, not necessarily strongly connected.
The following lemma allows us to characterize the linear complexity of termination for $\A$ by applying
Theorem~\ref{thm:precise-complexity} to its strongly connected components. A~proof is straightforward.
\begin{lemma}
	\label{lem-SCC}
	Let $d \in \N$, and let $\A = \ce{Q,T}$ be a $d$-dimensional VASS. Then $\term(n) \in  \bigO(n)$ iff $\term_R(n) \in \bigO(n)$ for every SCC $R$ of $Q$, where $\term_R(n)$ is the termination complexity of $\A_R$.
\end{lemma}

\begin{corollary}
	\label{coro:lin-polytime}
The problem whether the termination complexity of a given
   $d$-dimensional VASS is linear is solvable in time polynomial in the size
   of~$\A$.
\end{corollary}
\section{Polynomial termination time}
\label{sec-poly}

We now concentrate on VASS with polynomial termination complexity.
For simplicity, we restrict ourselves to \emph{strongly connected VASS}. The
general case is discussed at the end of the section.

A prominent notion in our analysis is the one of a \emph{ranking} function for
VASS.
Let $\A = \ce{Q,T}$ be a
VASS. A \emph{linear map} for $\A$ is a function $f$ assigning rational numbers
to
configurations of $\A$ s.t. there exists a vector $\bolden{c}_f$ and a
weighting
vector $\weights_f\in \Q^{\locs}$ such that for each configuration
$p\bv$ of $\A$ it holds $f(p\bv)=\bolden{c}_f^\top\cdot \bv + \weights_f(p)$. The
vector $\bolden{c}_f$ is called a \emph{normal} of $f$. Given
a linear map $f$, we say that a transition $(p,\bu,q)$ of $\A$ is $f$-ranked if
$\bolden{c}_f^\top \cdot \bu + \weights_f(q) \leq \weights_f(p) - 1$ and $f$-neutral if
$\bolden{c}_f^\top \cdot \bu + \weights_f(q) = \weights_f(p)$. A linear map $f$ is a
\emph{quasi-ranking} function (QRF) for $\A$ if $\bolden{c}_f \geq
\vec{0}$ and if all transitions of $\A$ are either
$f$-ranked or $f$-neutral, and a \emph{ranking} function (RF) if $\bolden{c}_f \geq
\vec{0}$ and all transitions of $\A$ are $f$-ranked.
A quasi-ranking function $f$ is positive if each component of
$\bolden{c}_f$ is positive. Note that in the language of update and incidence
matrices $\updates$ and $\flowMatrix$ the conditions can be phrased as follows:
a linear map
$f$ is a QRF if and only if $\bc_f \geq 0$ and $\bc_f^{\top}\cdot \updates
-\bw_f^\top\cdot \flowMatrix \leq \vec{0}^\top$ such that if there is a negative number in some column, it is $\leq -1$.
Similarly, a linear map
$f$ is a RF if and only if $\bc_f \geq 0$ and $\bc_f^{\top}\cdot \updates
-\bw_f^\top\cdot \flowMatrix \leq \vec{-1}^\top $.

The existence of ranking functions is already tightly connected to the question
whether a given VASS has linear complexity, as shown in the following theorem.

\begin{theorem}
\label{thm:onedim_rf}
A VASS $\A$ has a linear termination complexity if and only if there exists a
ranking function for $\A$.
\end{theorem}
\begin{proof}[Proof.]
Consider the LP $\rationalLP$ from Theorem~\ref{thm:precise-complexity}. Its
	dual LP is the LP $\rationalLPdual$ pictured in
	Figure~\ref{fig:dual-linear}.
\begin{figure}
	\begin{center}
		\begin{tabular}{|c|}
			\hline
			
			{\begin{minipage}[c]{0.5\linewidth}
					
					\vspace{0.2cm}
					$\min \by_U^\top \cdot \oneVec$  with
					\begin{align*}
					\by_{\updates}^\top \cdot \updates - \by_\flowMatrix^\top
					\cdot \flowMatrix & \leq -\oneVec^\top\\
					\by_\updates & \geq \zeroVec
					\end{align*}
				\end{minipage}}\\
				\hline
			\end{tabular}
		\end{center}

\caption{The rational LP $\rationalLPdual$ that is dual to $\rationalLP$. Here
the variables are vectors
$\by_\updates \in
\Q^d$ and
		$\by_\flowMatrix \in \Q^\locs$.}
\label{fig:dual-linear}
\end{figure}
		
		The dual LP has a feasible solution if and only if the original LP has
		an optimal solution (since it always has a feasible solution) and that
		is if and only if the VASS $\A$ is linear (due to
		Theorem~\ref{thm:precise-complexity}).
		Assume there exists a feasible solution. Let $f$ be a function such that
		\[
		f(p\bv) \quad=\quad \by_\updates^\top\cdot \bv + \by_\flowMatrix(p)
		\]
		i.e., $\bc_f = \by_\updates$ and $\weights_f(p) = \by_\flowMatrix(p)$.
		From the constraints of the dual LP we obtain for any transition
		$(p,\bu,q)$
		\[
		\bc_f^\top \cdot \bu - (\weights_f(p) - \weights_f(q)) \leq -1,
		\]
		i.e. $f$ is a RF.
		Conversely, let $f$ be any RF. Then $\by_\updates = \bc_f$,
		$\by_\flowMatrix = \weights_f$ is a feasible solution for the dual LP.
\end{proof}

Below, we show that complexity of general VASS $\A$ is highly influenced by
properties of
normals of QRFs for $\A$. In particular, we classify each VASS $\A$ into one of
three types:

\begin{compactitem}
	\item[(A)] Non-terminating VASS.
	\item[(B)] Positive normal VASS: Terminating VASS $\A$ for which there
	exists a QRF $f$ s.t. each
	component of the normal $\bc_f$ is positive.
	\item[(C)] Singular normal VASS: Terminating VASS $\A$ for which there
	exists a QRF $f$ for $\A$ s.t. each
	the normal $\bc_f$ is non-negative
	and~(B) does not hold.
\end{compactitem}

\textbf{Results.} We perform our complexity analysis on top of the above
classification. We show that each non-trivial
type (B)
VASS of dimension $d$ has termination complexity
 in $\Theta(n^k)$, where $1\leq k \leq d$ is an integer. Condition (C) is not
strong enough to guarantee polynomial termination complexity, and hence
singularities in the QRF normals are the key reason for complex asymptotic
bounds in VASS. On the algorithmic front, we present a polynomial-time algorithm
which classifies VASS into one of the above classes. Moreover, for type (B)
VASS the
algorithm also computes the degree $k$ such that the termination complexity of
the VASS is $\Theta(n^k)$. Hence, we give a complete complexity classification
of type (B)
VASS. For type (C) VASS, the algorithm returns a valid lower bound: a $k$
such that the termination complexity is $\Omega(n^k)$ (in general, such a bound
does not have to be tight). In the following, we first present the algorithm
and then formally state and prove its properties, which establish the above
results.

Theorem~\ref{thm:onedim_rf} gives complete classification of linear complexity VASS. Note that the ranking function doesn't have to be positive. The following lemma shows that every linear VASS is actually of type~(B).
\begin{lemma}
Let $\A$ be a VASS. There exists a ranking function for $\A$ if and only if there exists a positive ranking function for $\A$.
\end{lemma}
\begin{proof}
One direction is trivial. For the other, assume we have some ranking function $f$ for $\A$.
Then for any transition $t = (p,\bu,q)$ we have $\bc_f \cdot \bu + \bw_f(q) \leq \bw_f(p) - 1$.

Let $\epsilon > 0$ be such that every transition $(p,\bu,q)$ we have
$\vec{\epsilon} \cdot \bu \leq 1$
(there are only finitely many transitions so such $\epsilon$ must exist).
We define a linear map $g$ as follows
\[
\bc_g = 2\bc_f + \vec{\epsilon} \quad \text{and} \quad \bw_g = 2\bw_f.
\]
Then for any transition $(p,\bu,q)$ we have
\begin{align*}
\bc_g \cdot \bu + \bw_g(q) &= 2\bc_f \cdot \bu + \vec{\epsilon}\cdot \bu + 2\bw_f(q) \leq\\
& 2\bw_f(p) - 2 + \vec{\epsilon}\cdot \bu \leq 2\bw_f(p) - 1 = \bw_g(p) - 1.
\end{align*}

Therefore, $g$ is a positive RF.
\end{proof}

\begin{algorithm}[t]
	\SetKwInOut{Input}{input}\SetKwInOut{Output}{output}
	\SetKwFunction{dec}{Decompose}
	\SetKwProg{proc}{procedure}{}{}
	\DontPrintSemicolon
	
	\Input{A strongly connected $d$-dimensional VASS $\A=(\locs,\trans)$ with
	at
		least one transition.}
	\Output{A tuple $(k,\precise)\in
	\{1,2,\dots,d\}\times\{\mathit{true},\mathit{false}\}$, or
	``non-terminating''.}
	
	\lIf{$\exists$ positive QRF for $\A$}{$\precise \leftarrow \mathit{true}$}\label{alg-line-1}
	\lElse{$\precise \leftarrow \mathit{false}$}\label{alg-line-2}
	$k:=$\dec{$\A$}\label{alg-line3}\\
	\lIf{$k= \infty$}{\Return{''non-terminating''}}\lElse{
		\Return{($k,\precise$)}}
	\vspace{1mm}
	
	\proc{\dec{$\A$}}{
		{$f\leftarrow$ a QRF for $\A$ maximizing the no. of $f$-ranked
		transitions}\label{alg-line-7}\\
		$T_f\leftarrow \{\text{$f$-neutral transitions of $\A$ } \}$\label{alg-line-8}	\\
		\lIf{$T_f$ contains all transitions of
			$\A$}{\Return{$\infty$}}
		\lIf{$T_f = \emptyset$}{\Return{$1$}}\label{alg-line-10}
		$\A_1,\dots,\A_\ell \leftarrow \text{ all SCCs of $\A_{T_f} $ }$\label{alg-line-11}\\
		\Return{$1+\max($\dec{$\A_1$}$,\dots,$\dec{$\A_\ell$}$ ) $ }\label{alg-line-12}
	}
	
	\caption{Computing polynomial upper/lower bounds on the termination
		complexity of $\A$.}
	\label{algo:main}
\end{algorithm}

\textbf{Algorithm.} Our method is formalized in Algorithm~\ref{algo:main}. In
the algorithm, for a
VASS $\A = (Q,T)$ and $T' \subset T$, we denote by $\A_{T'} = (Q,T')$ a
pair obtained from $\A$ by removing all transitions not belonging to $T'$. Note that this may not be a VASS (since some state doesn't have to have an outgoing transition). An SCC of $\A_{T'}$ is a maximal strongly connected VASS in $\A_{T'}$. We
now formally state the properties of the algorithm, starting with bounds on its
running time.

\begin{theorem}
	\label{thm:termination-dimension}
	Algorithm~\ref{algo:main} runs in time polynomial in $\size{\vass}$.
	In particular, when called on a VASS of dimension $d$, the overall depth of
	recursion is $< d$.
\end{theorem}

We proceed with correctness of the algorithm w.r.t. non-termination.

\begin{theorem}
\label{thm:nonterm}
	Assume that on input $\A$, Algorithm~\ref{algo:main} returns
	``non-terminating.''
	Then $\A$ is a non-terminating VASS.
\end{theorem}

Finally, the following two theorems show the correctness of the algorithm w.r.t.
upper and lower bounds on the termination complexity of VASS.

\begin{theorem} \label{thm:O-bound}
	Assume that on input $\A$, Algorithm~\ref{algo:main} returns a tuple
	$(k,\precise)\in \N \times \{\mathit{true},\mathit{false}\}$. Then $k\in
	\{1,\dots,d\}$ and $\A$
	is terminating. Moreover, if $\precise = \mathit{true}$, then
	$\term(n) \in
	\bigO(n^k)$.
\end{theorem}

\begin{theorem}
\label{thm:omega-bound}
	Assume that on input $\A$, Algorithm~\ref{algo:main} returns a tuple
	$(k,\precise) \in \N \times \{\mathit{true},\mathit{false}\}$. Then $k\in
	\{1,\dots,d\}$ and
	$\term(n)\in \Omega(n^k)$.
\end{theorem}

Note that the algorithm indeed performs the required classification since
$\precise$ is set to $\mathit{true}$ if and only if the check for the existence of
a positive QRF in the beginning of the algorithm is successful. We now present
the proofs of the above theorems.

\textbf{Proof of Theorem~\ref{thm:termination-dimension}.}
In order to analyze the termination of the algorithm we consider the cone of
cycle effects.
As usual we define the dimension $\dimOp(C)$ of a cone $C$ as the dimension of
the smallest vector space containing $C$.
We show that the dimension of the cone generated by $\Inc_\vass$ decreases with
each recursive call:

\begin{lemma}
	\label{lem:recursive-call-dimension}
	Let $\vass$ be some VASS such that $\dec(\vass)$ leads to some recursive
	call $\dec(\vass')$ for some SCC $\vass'$ of $\vass$.
	Then $\dimOp(\cone(\Inc_\vass)) > \dimOp(\cone(\Inc_{\vass'}))$.
\end{lemma}

By Lemma~\ref{lem:recursive-call-dimension} we have that the dimension of
$\cone(\Inc_\vass)$ decreases with every recursive call.
With $\dimOp(\cone(\Inc_\vass)) \le d$, we get that the recursion depth is
bounded by $d - 1$.

Now we focus on the complexity of computing a QRF $f$ maximizing the number of
$f$-ranked transitions. The computation of such a QRF can be directly encoded
by the following linear optimization problem $\qrfLP$.
\vspace{0.1em}
\begin{center}
\begin{tabular}{|c|}
	\hline
	{\begin{minipage}[c]{0.6\linewidth}
			\vspace{0.2cm}
			LP ($\qrfLP$): \hspace{0.3cm} $\max \oneVec^\top \activeQ$
			\begin{align*}
			\zeroVec & \le \activeQ \le \oneVec\\
			\rankCoeff & \ge \zeroVec \\
			\rankCoeff^\top \cdot \updates - \offsets^\top \cdot \flowMatrix & \le -\activeQ^\top
			\end{align*}
		\end{minipage}}\\
		\hline
	\end{tabular}
\end{center}
\vspace{0.2em}

	\begin{lemma}
		\label{lem:lp-qrf}
		Let $\rankCoeff,\offsets,\activeQ$ be an optimal solution to LP
		($\qrfLP$).
		Then, $\qrf(\loc \val) = \rankCoeff^\top \val + \offsets(\loc)$ is a
		QRF, which is maximizing the number of $\qrf$-ranked transitions.
	\end{lemma}

Similarly, checking the existence of a positive QRF can be performed by a
direct reduction to linear programming. The LP is analogous to $\qrfLP$.

\begin{lemma}
\label{lem:positive-QRF-check}
Checking the existence of a positive QRF can be done in polynomial time.
\end{lemma}

We now finish the proof of Theorem~\ref{thm:termination-dimension}.
We note that computing the QRFs in the algorithm can be done by linear
programming.
We next consider the set of recursive calls made at recursion depth $i$.
The VASSs of these recursive calls are all disjoint sub-VASSs of $\vass$.
Thus, the complexity of solving all the optimization problems at level $i$ is
bounded by the complexity of solving $\qrfLP$ for $\vass$.
Hence, the overall complexity of $\dec(\vass)$ is the complexity of solving
$\qrfLP$ times the dimension $d$.

\textbf{Proof of Theorem~\ref{thm:nonterm}.}
\newcommand{\conSysA}{\ensuremath{A}}
\newcommand{\conSysB}{\ensuremath{B}}
\newcommand{\character}{\ensuremath{\mathit{char}}}
\newcommand{\identity}{\ensuremath{\mathbf{Id}}}
Let $\vass = \ce{\locs,\trans}$ be a VASS.
Consider the constraint systems ($\conSysA_\transition$) and
($\conSysB_\transition$) stated below.
Both constraint systems are parameterized by a transition $\transition \in T$.
Constraint system ($\conSysA_\transition$) is taken from Kosaraju and
Sullivan~\cite{KS88}. Note that system $\conSysA_\transition$ has a rational solution if and only if it has an integer solution.

\vspace{0.2cm}
\begin{tabular}{|c|c|}
	\hline
	{\begin{minipage}[c]{0.4\linewidth}
			\vspace{0.2cm}
			constraint system ($\conSysA_\transition$):
			\begin{align}
			\updates \counters & \ge \zeroVec \label{multcycle:eq1}\\
			\counters & \ge \zeroVec \label{multcycle:eq2} \\
			\flowMatrix \counters & = \zeroVec \label{multcycle:eq3}\\
			\counters(\transition) & \ge 1 \label{multcycle:eq4}
			\end{align}
			\vspace{-0.4cm}
		\end{minipage}}
		&
		{\begin{minipage}[c]{0.5\linewidth}
				\vspace{0.2cm}
				constraint system ($\conSysB_\transition$):
				\begin{align*}
				\rankCoeff & \ge  \zeroVec\\
				\rankCoeff^\top \cdot \updates - \offsets^\top \cdot \flowMatrix & \le
				\zeroVec^\top \text{ with } -1 \\
				& \quad \text{ in column } \transition
				\end{align*}
			\end{minipage}}\\
			\hline
		\end{tabular}
		\vspace{0.2cm}
		
		The next lemma shows the connection between $\conSysA_\transition$ and multi-cycles in $\A$.
		We call a multi-cycle $\multCycle$ \emph{non-negative} if $\eff(\multCycle)\geq \vec{0}$.
		
		\begin{lemma}[Cited from~\cite{KS88}]
			\label{lem:solution-is-multcycle}
			There is a solution $\counters \in \mathbb{Z}^{T}$  to constraints
			(\ref{multcycle:eq1})-(\ref{multcycle:eq3}) iff there exists a
			non-negative multi-cycle $\multCycle$ such that the number of times a transition $t$ appears in cycles of $\multCycle$ is at least $\counters(\transition)$, for each $\transition \in T$.
		\end{lemma}
	
	On the other hand, the system $\conSysB_\transition$ is connected to QRFs.
	
	\begin{lemma}
\label{lem:constr-to-QRF}
Constraint system ($\conSysB_\transition$) has a rational solution $\bc,\bw$ if and only if there exists a $k \in \R^+$ and a QRF
$\qrf$ with $\bc_\qrf = k \cdot \bc$ and $\bw_\qrf = k\cdot \bw$ such that transition $\transition$ is $\qrf$-ranked and every other
transition is $\qrf$-ranked or $\qrf$-neutral.		
	\end{lemma}
	
The following result is an immediate consequence of Farkas' Lemma.

\begin{lemma}
	\label{lem:ranking-or-witness}
	For each $t\in\trans$ exactly one of the constraint systems ($\conSysA_\transition$) and
	($\conSysB_\transition$) has a solution.
\end{lemma}

We now finish the proof of Theorem~\ref{thm:nonterm}.
Because Algorithm~\ref{algo:main} returns ``non-terminating'', there is a sub-VASS $\vass'$ of $\A$, encountered during some recursive call, such that no transition of $\A$ is $\qrf$-ranked for any QRF $\qrf$.
Hence, constraint system $\conSysB_\transition$ is unsatisfiable for every transition $\transition$ of $\vass'$.
By Lemma~\ref{lem:ranking-or-witness}, constraint system $\conSysA_\transition$ is satisfiable.
We consider the non-negative multi-cycle $\multCycle$ associated to an integer solution of $\conSysA_\transition$.
This multi-cycle contains at least transition $\transition$.
Because such a multi-cycle exists for every transition $\transition$, we can
combine all these multi-cycles into a single non-negative cycle, which shows
that $\A$ is non-terminating.

\paragraph{Connection to~\cite{KS88}.}
Algorithm~\ref{algo:main} extends algorithm ZCYCLE of Kosaraju and Sullivan~\cite{KS88} by a ranking function construction.
Because of the duality stated in Lemma~\ref{lem:ranking-or-witness},
the ranking function construction part can be interpreted as the \emph{dual} of algorithm ZCYCLE.
Algorithm~\ref{algo:main} makes use of this duality to achieve completeness:
it either returns a ranking function, which witnesses termination, or it returns a non-negative cycle, which witnesses non-termination.
The duality also means that ranking function construction comes essentially for free, as primal-dual LP solvers simultaneously generate solutions for both problems.
An additional result is the improved analysis of the recursion depth:
\cite{KS88} uses the fact that the number of locations $|Q|$ is a trivial upper bound of the recursion depth,
while we have shown the bound $\dimOp(\vass)+1$ (see Theorem~\ref{thm:termination-dimension}).
With this result and with LP ($\qrfLP$), which simultaneously solves all
constraint systems ($\conSysA_\transition$)/($\conSysB_\transition$) and thus
avoids an iteration over $\transition$, we affirmatively answer the open
question of Kosaraju and Sullivan~\cite{KS88}, whether the
complexity can be expressed as a polynomial function in the dimension $d$ times
the complexity of a linear program.

\textbf{Proof of Theorem~\ref{thm:O-bound}.}
First, we will prove the $\bigO$-bound by induction on the depth of recursion (of \dec). More precisely, if the algorithm return $tight = \mathit{true}$ and the depth of recursion (number of calls) is $i$, the termination complexity is in $\bigO(n^{i+1})$.
\begin{compactitem}
	\item If there is no recursive call of procedure \dec then QRF $f$ obtained 
	on line~\ref{alg-line-7} is actually a RF, because $T_f = \emptyset$ i.e., 
	all transitions are $f$-ranked. Due to Theorem~\ref{thm:onedim_rf} we have 
	$\term_\A \in~\bigO(n)$.
	
	\item	Let $i>0$ be the recursion depth. Assume the claim is correct for every run of the algorithm with recursion depth $<i$. By induction hypothesis we have that every SCC $\A_j$ of $\A_{T_f}$ has termination complexity $\term_{\A_j} \in \bigO(n^i)$.
	
	Let $q_0\bu_0$ be an initial configuration. Now assume we have a VASS $\A$ and QRF $f$. If a transition is $f$-ranked, the $f$-value of the next configuration decreases by at least 1. If it is $f$-neutral, it does not increase. Notice that every configuration $p\bv$ satisfies $f(p\bv) = \bc_f^\top \cdot \bv + \weights_f(p) \geq \weights_f(p) \geq \min_{q \in \locs} \weights_f(q)$ since $\bv$ and $\bc_f$ are non-negative. Therefore, any zero-avoiding path can have at most $f(q_0\bu_0) - \min_{q \in \locs} \weights_f(q)$ of $f$-ranked transitions.
	
	Let $g$ be the positive QRF whose existence is ensured on line~\ref{alg-line-1} of the algorithm (since algorithm returns $(k,\mathit{true})$). We give a linear bound on the size of counters in every configuration $p\bv$ reachable from $q_0\bu_0$. Since $g$ is a QRF, we have
	\begin{align*}
	g(q_0\bu_0) &\geq g(p\bv),\\
	\weights_g(q_0) + \sum_{i=1}^d \bc_g(i) \cdot \bu_0(i) &\geq \weights_g(p) + \sum_{i=1}^d \bc_g(i) \cdot \bv(i).
	\end{align*}
	Let $c_{max} = \max_{i = 1,\dots,d} \bc_g(i)$ and $c_{min} = \min_{i=1,\dots,d} \bc_g(i)$. Now for any counter $j \in \{1,\dots,d\}$ we have the following estimates.
	\begin{align*}
	\weights_g(q_0) - \weights_g(p) + \sum_{i=1}^d c_{max} \cdot \bu_0(i) &\geq c_{min} \bv(j),\\
	\frac{\weights_g(q_0) - \min_{q \in \locs} \weights_g(q)}{c_{min}} + d \cdot \frac{c_{max}}{c_{min}} \cdot \max_{i=1,\dots,d} \bu_0(i) &\leq \bv(j).
	\end{align*}
	
	Therefore, the size of any reachable configuration is
	linearly bounded by the size of the initial configuration and after $\bigO(n^i)$ transitions in some SCC we have to do at least one $f$-ranked transition. From this we obtain $\term(\A) =
	\bigO(n^{i+1})$ and the proof is done.
\end{compactitem}
Now we want to prove that the VASS terminates even if $\precise = \mathit{false}$. Again, we do the proof by induction on the depth of recursion of \dec. The base of the induction is the same as in the proof of $\bigO$-bound. In the induction step we assume only that every SCC of $\A_{T_f}$ is terminating. Again, no transition increases the $f$-value and we can do only $f(q_0\bu_0) - \min_{q \in \locs} \weights_f(q)$ of $f$-ranked transitions, therefore $\A$ is terminating (we cannot stay in one SCC indefinitely and by switching between them, we have to make at least one $f$-ranked transition).

\textbf{Proof of Theorem~\ref{thm:omega-bound}. } We now prove the correctness
of our algorithm w.r.t. lower bounds. To do this, we show how to construct, for
each sufficiently large $n$, a path of length $\Omega(n^k)$ which results into
a computation. We start with a lemma which shows a useful
property of a QRF
$f$ that maximizes the number of $f$-ranked transitions: whenever we have a
cycle consisting solely of $f$-neutral transitions, the effect of this cycle
can be in some sense compensated by executing a combination of some other
cycles.

\begin{lemma}\label{lem:goodQRF}
	Let $\A$ be a connected VASS, and let $f$ be a QRF for $\A$ which maximizes the number of $f$-ranked transitions.
	Let $\bc_f$ be the normal of $f$.
	Then for each vector $\bv \in \Inc$ with $\bc_f^\top \cdot \bv = 0$ there exists a vector $\bw \in \cone(\Inc)$ such that $\bv + \bw \geq \vec{0}$.
\end{lemma}

We now proceed with the proof of Theorem~\ref{thm:omega-bound}. We show that if
$\A$ is a strongly connected VASS, and the call $\texttt{Decompose}(\A)$
returns a number $k\in \N$, then for all sufficiently large $n\in\N$ there
exists a configuration $p_n \vec{n}$ and a computation $\beta_n$
of length at least $b\cdot n^k$ initiated in $p_n\vec{n}$, where $b\in \Q^+$ is
a fixed
positive constant independent of $n$. We proceed by induction on $k$. If $k=1$,
then $\A$ admits a RF and the existence of such a zero-avoiding computation of
linear length follows from Theorems~\ref{thm:onedim_rf}
and~\ref{thm:precise-complexity}.  Now assume that $k>1$. Then the call
$\texttt{Decompose}(\A)$ must result in a recursive sub-call
$\texttt{Decompose}(\A')$ which returns $k-1$. We prove that for all
sufficiently large $n$ there
exists a computation initiated in some $p_n\vec{n}$ of length
$\Omega(n^k)$. We
prove the existence of such a path in several sub-steps.

 \emph{Constructing the paths of length $\Omega(n^{k-1})$.}\quad Since the
 termination complexity of $\A'$ is $\Omega(n^{k-1})$ (by induction
 hypothesis), there is $b \in
 \R^+$ such that for all sufficiently large $m \in \N$ there exist a
 configuration $p_m \vec{m}$ and a computation $\beta_m$ of
 length at least $b \cdot m^{k-1}$ initiated in $p_m \vec{m}$. Since
 $\pi_{\beta_m}$ inevitably contains a cycle whose length is at least $b' \cdot
 m^{k-1}$ (for some fixed $b' \in \R^+$ independent of $\beta_m$), we can
 safely
 assume that $\pi_{\beta_m}$ is actually a cycle, which implies
 $\eff(\pi_{\beta_m}) \in \cone(\Inc)$.

 \emph{Constructing the compensating path.}\quad Since $\pi_{\beta_m}$ is such
 that $\bc_f^\top \cdot \eff(\pi_{\beta_m}) = 0$ and $\eff(\pi_{\beta_m}) \in
 \cone(\Inc)$, it follows from Lemma~\ref{lem:goodQRF} that there exists $\bu
 \in \cone(\Inc)$ such that $\bu + \eff(\pi_{\beta_m}) \geq \vec{0}$ i.e., $\bu \geq
 -\eff(\pi_{\beta_m})$. Since $\bu = \sum_{j=1}^k a_j \cdot \bv_j$, where $k
 \in \N$, $a_j \in \Q^+$, and $\bv_j \in \Inc$ for all $1 \leq j \leq k$, a
 straightforward idea is to define the compensating path by ``concatenating''
 $\lfloor a_j \rfloor$ copies of $\gamma_j$, where $\eff(\gamma_j) = \bv_j$,
 for all $1 \leq j \leq k$. This would produce the desired effect on the
 counters, but there is no bound on the counter decrease in intermediate
 configurations visited when executing this path. To overcome this problem, we choose $m$ and
 construct the compensating path for $\pi_{\beta_m}$ more carefully. We use the following lemma.

\begin{lemma}\label{lem-compens}
Let $\A$ be a VASS and $f$ a QRF maximizing the number of $f$-ranked transitions. Then there exists $\delta \in \R^+$ such that for every $m \in \N$ and every cycle $\pi_m$ with $\bc_f^\top \cdot \eff(\pi_m) = 0$ and $\eff(\pi_m)\geq -\vec{m}$ there is a path $\varrho_m$ such that $\eff(\pi_m) + \eff(\varrho_m) \geq -(d+1)\cdot |Q|\cdot\maxup{\A}$ and no counter is decreased by more than $\delta \cdot m$ along $\varrho_m$.
\end{lemma}

 \emph{Constructing a computation $\alpha_n$ of length
 $\Omega(n^k)$.}\quad Now we are ready to put the above ingredients together,
 which still requires some effort.
 
 Assume $p\bv$ is an initial configuration. Now we need only $|Q|$ transitions in order to get to the SCC $\A'$ where we execute a path $\pi_{\beta_m}$ with $\bc_f^\top \cdot \eff(\pi_{\beta_m}) = 0$ of length $\Omega(n^{k-1})$.
 
 We need to choose $m$ as large as possible but small enough so that we can execute path $\pi_{\beta_m}$ and its compensating path $\varrho_m$. At the end of $\pi_{\beta_m}$ some counters may be decreased by $m$ to $n-(|Q|\cdot\maxup{\A})-m$ (remember, we used at most $|Q|$ transitions to get to the starting state of $\pi_{\beta_m}$). Then we need to execute the compensating path $\varrho_m$. For this we need counters of size at most $\delta \cdot m$ + $|Q|\cdot \maxup{\A}$ (we need to reach the initial state of $\varrho_m$ and then execute this path). Together we need
\[
n-|Q|\cdot\maxup{\A}-m \geq \delta \cdot m + |Q| \cdot \maxup{\A}.
\]
We want to maximize $m$ (in order to get a long path). This yields
\[
m = \left\lfloor\frac{n - 2|Q|\cdot \maxup{\A}}{1+\delta}\right\rfloor.
\]
After this, every counter decreased by at most $(d+3)|Q|\cdot\maxup{\A}$ (we needed to get to the right SCC and then run the compensating path).

Repeating this procedure $\bigO(n)$ times, we obtain a path of length $\Omega(n \cdot n^{k-1}) = \Omega(n^k)$. This finishes the proof of Theorem~\ref{thm:omega-bound}.

\textbf{Non-strongly connected VASS.} We finish our analysis with a remark that
our complete complexity classification of type (B) VASS extends to non-strongly
connected VASS whose each SCC is also of type (B).

\begin{lemma}
Let $\A$ be a VASS, $\A_1,\dots,\A_l$ its SCCs (reachable from the initial configuration).
\begin{enumerate}
\item $\term_{\A} \in \Omega(\max_{i \in \{1,\dots,l\}} \term_{\A_l})$
\item If for every SCC of $\A$ there is a positive QRF then $\term_{\A} \in \bigO(\max_{i \in \{1,\dots,l\}} \term_{\A_l})$.
\end{enumerate}
\end{lemma}
\begin{proof}
The first part of the lemma is trivial. Since we can visit any SCC $\A_i$ in a number of transitions bounded by $|\locs|$ from the initial one, the asymptotic complexity cannot be lower than that of $\A_i$.

As in the proof of Theorem~\ref{thm:O-bound} we have for any SCC $\A_i$ and any configuration $q\bu$ with $q \in \locs_{\A_i}$ that the size of any $p\bv$ with $p \in \locs_{\A_i}$ reachable from $q\bu$ is linearly bounded by some constant depending only on the positive QRF for $\A_i$.

Since the number of SCCs for a given VASS is a constant, the size of the counter vector can increase during any computation at most by a factor independent of the size of the initial configuration. Therefore, the second claim holds.
\end{proof}
\section{Conclusions}
\label{sec-concl}

Our result gives rise to a number of interesting directions for future work.
First, whether our precise complexity analysis or the complete method can be extended to other
models in program analysis (such as affine programs with loops) is an interesting theoretical
direction to pursue. 
Second, our result can be used for developing a scalable tool for sound and 
complete analysis of asymptotic bounds for VASS.


\clearpage
\appendix
\begin{center}
	{\Large Technical Appendix}
\end{center}

\section{Proof of Lemma~\ref{lem-no-halfspace}}

	We distinguish two possibilities.
	\begin{itemize}
		\item[(a)] There exists a closed half-space $\hat{\calH}_{\bn}$ of 
		$\R^d$
		such that $\bn > \vec{0}$ and  $X \subseteq \hat{\calH}_{\bn}$.
		\item[(b)] There is no closed half-space $\hat{\calH}_{\bn}$ of $\R^d$ 
		such 
		that $\bn > \vec{0}$ and  $X \subseteq \hat{\calH}_{\bn}$.
	\end{itemize}
	\textbf{Case~(a).} We show that there exists $\bu \in X$ such that $\bu 
	\neq 
	\vec{0}$ and $-\bu \in \Cone(X)$. Note that this immediately implies the 
	claim of our lemma---since $-\bu \in \Cone(X)$, there are 
	$\bv_1,\ldots,\bv_k \in X$ and $c_1,\ldots,c_k \in \R^+$ such that $-\bu = 
	\sum_{i=1}^k c_i \bv_i$. Since all elements of $X$ are vectors of 
	non-negative integers, we can safely assume $c_i \in \Q^+$ for all $1\leq i 
	\leq k$. Let $b$ be the least common multiple of $c_1,\ldots,c_k$. Then 
	$b\bu + (b\cdot c_1) \bv_1 + \cdots + (b\cdot c_k) \bv_k = \vec{0}$ and we 
	are done.
	
	It remains to prove the existence of $\bu$. Let us fix a normal vector $\bn 
	> \vec{0}$ such that $X \subseteq \hat{\calH}_{\bn}$ and the set $X_{\bn} = 
	\{ \bv \in X \mid \bv^\top \cdot \bn < 0\}$ is \emph{maximal} (i.e., there is no 
	$\bn' > \vec{0}$ satisfying $X \subseteq \hat{\calH}_{\bn'}$ and $X_{\bn} 
	\subset X_{\bn'}$). Further, we fix $\bu \in X$ such that $\bu^\top \cdot \bn = 
	0$. Note that such $\bu \in X$ must exist, because otherwise $X_{\bn} = X$ 
	which contradicts the assumption of our lemma. We show $-\bu \in \Cone(X)$. 
	Suppose the converse. Then by Farkas' lemma there exists a separating 
	hyperplane for $\Cone(X)$ and $-\bu$ with normal vector $\bn'$, i.e., $\bv^\top 
	\cdot \bn' \leq 0$ for all $\bv \in X$ and $-\bu^\top \cdot \bn' > 0$. Since 
	$\bn 
	> \vec{0}$, we can fix a sufficiently small $\varepsilon > 0$ such that the 
	following conditions are satisfied:
	\begin{itemize}
		\item $\bn + \varepsilon \bn' > \vec{0}$,
		\item for all $\bv \in X$ such that $\bv^\top \cdot \bn < 0$ we have that 
		$\bv^\top 
		\cdot (\bn + \varepsilon \bn') < \vec{0}$.
	\end{itemize}
	Let $\bw = \bn + \varepsilon \bn'$. Then $\bw > 0$, $\bv^\top \cdot \bw < 0$ for 
	all $\bv \in X_{\bn}$, and $\bu^\top \cdot \bw = \bu^\top \cdot \bn + \varepsilon 
	(\bu^\top 
	\cdot \bn') = \varepsilon (\bu^\top \cdot \bn') < 0$. This contradicts the 
	maximality of $X_{\bn}$.
	\smallskip
	
	\noindent
	\textbf{Case~(b).} Let $B = \{\bv \in \R^d \mid \bv \geq \vec{0} \text{ and 
	}  1  \leq \sum_{i=1}^d \bv(i)  \leq  2\}$. We prove $\Cone(X) \cap B \neq 
	\emptyset$, which implies the claim of our lemma (there are 
	$\bv_1,\ldots,\bv_k \in X$ and $c_1,\ldots,c_k \in \Q^+$ such that 
	$\sum_{i=1}^k c_i \bv_i \in B$). Suppose the converse, i.e., $\Cone(X) \cap 
	B = \emptyset$. Since both $\Cone(X)$ and $B$ are closed and convex and $B$ 
	is also compact, we can apply the ``strict'' variant of hyperplane 
	separation theorem. Thus, we obtain a vector $\bn \in\R^d$ and a constant 
	$c 
	\in \R$ such that $\bx^\top \cdot \bn < c$ and $\by^\top \cdot \bn > c$ for all $\bx 
	\in \Cone(X)$ and $\by \in B$. Since $\vec{0} \in \Cone(X)$, we have that 
	$c 
	> 0$. Further, $\bn \geq \vec{0}$ (to see this, realize that if $\bn(i) < 
	0$ 
	for some $1 \leq i \leq d$, then $\by \cdot \bn < 0$ where $\by(i) = 1$ and 
	$\by(j) = 0$ for all $j \neq i$; since $\by \in B$ and $c > 0$, we have a 
	contradiction). Now we show $\bx^\top \cdot \bn \leq 0$ for all $\bx \in 
	\Cone(X)$, which contradicts the assumption of Case~(b). Suppose $\bx^\top \cdot 
	\bn > 0$ for some $\bx \in \Cone(X)$. Then $(m \cdot \bx)^\top \cdot \bn > c$ 
	for 
	a sufficiently large $m \in \N$. Since $m \cdot \bx \in \Cone(X)$, we have 
	a 
	contradiction.
	
\section{Proof of Corollary~\ref{coro:lin-polytime}}

If the VASS is not strongly connected, then, due to Lemma~\ref{lem-SCC}, it 
suffices to check linearity of the termination for all its strongly connected 
components.
Consider a strongly connected VASS. The~LP~(R) always has a feasible solution 
(e.g. a zero 
vector). Hence, checking whether the complexity is linear amounts to checking 
whether the LP has a bounded objective function, which can be done in 
polynomial time.
\section{Proof of Theorem~\ref{thm:onedim_rf}}
Consider the LP from Theorem~\ref{thm:precise-complexity}.
\begin{center}
\begin{tabular}{|c|}
\hline

{\begin{minipage}[c]{0.5\linewidth}
\vspace{0.2cm}

rational LP ($\rationalLP$):

\vspace{0.2cm}
$\max \oneVec^\top\countersAlt$  with
\begin{align*}
  \countersAlt & \ge \zeroVec\\
  \updates \countersAlt & \ge -\oneVec\\
  \flowMatrix \countersAlt & = \zeroVec
\end{align*}
\end{minipage}}\\
\hline
\end{tabular}
\end{center}
The dual LP is
\begin{center}
\begin{tabular}{|c|}
\hline

{\begin{minipage}[c]{0.5\linewidth}

\vspace{0.2cm}
$\min \by_U^\top \cdot \oneVec$  with
\begin{align*}
  \by_{\updates}^\top \cdot \updates - \by_\flowMatrix^\top \cdot \flowMatrix & \leq -\oneVec^\top\\
  \by_\updates & \geq \zeroVec
\end{align*}
\end{minipage}}\\
\hline
\end{tabular}
\end{center}
Here the variables are vectors $\by_\updates \in \Q^d$ and $\by_\flowMatrix \in \Q^\locs$.

The dual LP has a feasible solution if and only if the original LP has an optimal solution (since it always has a feasible solution) and that is if and only if the VASS $\A$ is linear (due to Theorem~\ref{thm:precise-complexity}).

Assume there exists a feasible solution. Let $f$ be a function such that
\[
f(p\bv) \quad=\quad \by_\updates^\top\cdot \bv + \by_\flowMatrix(p)
\]
i.e., $\bc_f = \by_\updates$ and $\weights_f = \by_\flowMatrix$. From the constraints of the dual LP we obtain for any transition $(p,\bu,q)$
\[
\bc_f^\top \cdot \bu - (\weights_f(p) - \weights_f(q)) \leq -1,
\]
i.e. $f$ is a RF.

Conversely, let $f$ be any RF. Then $\by_\updates = \bc_f$, $\by_\flowMatrix = 
\weights_f$ is a feasible solution for the dual LP.

\section{Proof of Lemma~\ref{lem:recursive-call-dimension}}

	Clearly, $\Inc_{\vass'} \subseteq \Inc_\vass$ because $\vass'$ is a
	sub-VASS of $\vass$ and hence $\dimOp(\cone(\Inc_\vass)) \ge
	\dimOp(\cone(\Inc_{\vass'}))$.
	Let $\qrf$ be the QRF computed for $\vass$ and let $\normalVec_\qrf$ be the
	associated normal.
	We will show that $\cone(\Inc_{\vass'})$ is contained in the hyperplane
	$\{\val \mid \normalVec_\qrf^\top \val = 0\}$ while $\cone(\Inc_\vass)$ is
	not.
	This is sufficient to infer $\dimOp(\cone(\Inc_\vass)) >
	\dimOp(\cone(\Inc_{\vass'}))$.
	
	We consider some $\val \in \Inc_{\vass'}$.
	We have $\valueSum(\cycle) = \val$ for some simple cycle $\cycle$.
	We consider the edges along the cycle $\cycle$.
	Because $\vass'$ appears in some recursive call of $\dec(\vass)$, we have
	that every edge of $\cycle$ is $\qrf$-neutral, i.e.,
	$\normalVec_\qrf^\top \update + \weights_\qrf(\loc') = \weights_\qrf(\loc)$
	for every transition $(\loc,\update,\loc')$ of $\cycle$.
	Adding these equations along the cycle $\cycle$ establishes $\val =
	\valueSum(\cycle) = 0$.
	Hence, $\cone(\Inc_{\vass'}) \subseteq \{\val \mid \normalVec_\qrf^\top
	\val = 0\}$.
	
	On the other hand, because there is a recursive call i.e., $T_f \neq T$, there is at least one
	transition of $\vass$ which is $\qrf$-ranked.
	Because $\vass$ is connected we can choose some simple cycle $\cycle$ which
	contains an $\qrf$-ranked transition.
	Adding the inequalities $\normalVec_\qrf^\top \update +
	\weights_\qrf(\loc') \le \weights_\qrf(\loc)$ for every transition
	$(\loc,\update,\loc')$ of $\cycle$ establishes $\normalVec_\qrf^\top
	\valueSum(\cycle) < 0$.
	Hence, $\cone(\Inc_\vass)$ is not contained in $\{\val \mid
	\normalVec_\qrf^\top \val = 0\}$.

\section{Proof of Lemma~\ref{lem:lp-qrf}}
	We state the following properties about LP ($\qrfLP$), which are easy to
	verify:
	LP ($\qrfLP$) is always satisfiable (consider $\rankCoeff = \zeroVec$,
	$\offsets = \zeroVec$ and $\activeQ = \zeroVec$).
	Let $\rankCoeff,\offsets,\activeQ$ and $\rankCoeff',\offsets',\activeQ'$ be
	feasible points of LP ($\qrfLP$).
	Then, (1) $\rankCoeff + \rankCoeff',\offsets +
	\offsets',\max\{\activeQ,\activeQ'\}$ is a feasible point of LP ($\qrfLP$)
	and (2) $d\rankCoeff,d\offsets,\min\{d\activeQ,\oneVec\}$ is a feasible
	point of LP ($\qrfLP$) for all $d \in \mathbb{Q}$ with $d \ge 1$.
	
			We now show that every transition is either $\qrf$-ranked or
			$\qrf$-neutral.
			It is sufficient to show that $\activeQ(\transition) = 0$ or
			$\activeQ(\transition) = 1$ for each transition $\transition$.
			Assume $0 < \activeQ(\transition) < 1$ for some transition
			$\transition$.
			Then, we can choose $d = \frac{1}{\activeQ(\transition)} > 1$ and 
			apply
			(2) in order to obtain a feasible point with value $\oneVec^\top
			d\activeQ > \oneVec^\top \activeQ$, which is a contradiction to the
			assumption that $\rankCoeff,\offsets,\activeQ$ is optimal.
			
			Assume that there is another QRF $\qrf'(\loc \val) = 
			\rankCoeff'^\top
			\val + \offsets'(\loc)$ and a transition $\transition$ such that
			$\transition$ is $\qrf'$-ranked but $\qrf$-neutral (we note that we
			must have $\activeQ(t) = 0$).
			We set $\activeQ'(\transition) = 1$ and $\activeQ'(\transition') = 
			0$
			for all transitions $\transition' \neq \transition$.
			We can now apply (1) in order to obtain a feasible point with value
			$\oneVec^\top \max\{\activeQ,\activeQ'\} > \oneVec^\top \activeQ$,
			which is a contradiction to the assumption that
			$\rankCoeff,\offsets,\activeQ$ is optimal.

\section{Proof of Lemma~\ref{lem:positive-QRF-check}}
Consider the following LP similar to the dual LP in the proof of Theorem~\ref{thm:onedim_rf}.
\begin{center}
\begin{tabular}{|c|}
\hline

{\begin{minipage}[c]{0.5\linewidth}
\vspace{0.2cm}
$\max \varepsilon$  with
\begin{align*}
  \by_{\updates}^\top \cdot \updates - \by_\flowMatrix^\top \cdot \flowMatrix & \leq \zeroVec^\top\\
  \by_\updates & \geq \vec{\varepsilon}
\end{align*}
\end{minipage}}\\
\hline
\end{tabular}
\end{center}
Here the variables are $\by_{\updates}$, $\by_{\flowMatrix}$ and $\varepsilon$. We show that this program has a feasible positive solution if and only if there exists a positive QRF.

If the program has a feasible positive solution then define a linear map $g$ such that $\bc_g = \by_\updates$ and $\weights_g = \by_\flowMatrix$. For every transition $(p,\bu,q)$ we have
\[
\bc_g^\top \cdot \bu - (\weights_g(p) - \weights_g(q)) \leq 0.
\]
Every transition $t = (p,\bu,q)$ that is not $g$-neutral, we have
\[
\bc_g^\top \cdot \bu - (\weights_g(p) - \weights_g(q)) = -\delta_t.
\]
Let $\delta = \min_{t \in \trans} \delta_t$. Then function $f = \frac{g}\delta$ is a positive QRF.

Positive QRF is a feasible solution for the dual program of 
Theorem~\ref{thm:onedim_rf}. It is therefore a feasible positive solution for 
this LP (since these constraints are weaker).

\section{Proof of Lemma~\ref{lem:constr-to-QRF}}

We know that a linear map $f$ is a QRF if and only if $\bc_f \geq 0$ and 
$\bc_f^{\top}\cdot \updates  -\bw_f^\top\cdot \flowMatrix \leq \vec{0}^T $ such 
that each column contains either 0 or a number $\leq -1$. Transitions with 
column $\leq -1$ are $f$-ranked, the rest are $f$-neutral. Therefore, any QRF 
$f$ with $(\bc_f^T \cdot U) (t) - (\bw_f^T \cdot F) (t) \leq -1$ satisfies the 
constraints.

Now let $\bc, \bw$ be a rational solution of $\conSysB_\transition$. Let
\[
\bu^T = \bc_f^{\top}\cdot \updates  -\bw_f^\top\cdot \flowMatrix \leq \vec{0}^T.
\]
Surely $\bu(t) \leq -1$. But it may happen that $0>\bu(t') >-1$ for some $t' 
\in T$. To remedy this, consider $k = \max_{\hat{t} \in T} \bu(\hat{t})$. Let
\[
\bc_f = \frac{1}{|k|}\bc, \quad \bw_f = \frac{1}{|k|}\bw.
\]
Now if $\bu(t') < 0$, then $\frac{1}{|k|}\bu(t') \leq -1$. Therefore, $f$ is a 
QRF such that $t$ is $f$-ranked.

\section{Proof of Lemma~\ref{lem:ranking-or-witness}}

We will use the following variant of \emph{Farkas' Lemma}, which states that 
given matrices $A$,$C$ and vectors $b$,$d$, exactly one of the following 
statements is true:

\vspace{0.2cm}
\begin{tabular}{|c|c|}
	\hline
	\begin{minipage}[c]{0.4\linewidth}
		constraint system ($A$):
		
		\vspace{0.2cm}
		there exists $\bolden{x}$ with
		
		\vspace{0.2cm}
		$\begin{array}{rcr}
		A\bolden{x} & \ge & \bolden{b} \\
		C\bolden{x} & = & \bolden{d}
		\end{array}$
		
	\end{minipage}
	&
	\begin{minipage}[c]{0.5\linewidth}
		\vspace{0.2cm}
		constraint system ($B$):
		
		\vspace{0.2cm}
		there exist $\bolden{y},\bolden{z}$ with
		
		\vspace{0.2cm}
		$\begin{array}{rcl}
		\bolden{y} & \ge & \zeroVec \\
		\by^\top \cdot A + \bz^\top \cdot C & = & \zeroVec^T \\
		\by^\top \cdot \bolden{b} + \bz^\top \cdot\bolden{d} & > & 0
		\end{array}$
		\vspace{0.1cm}
	\end{minipage} \\
	
	\hline
\end{tabular}
\vspace{0.2cm}

We fix some transition~$\transition$.
We denote by $\character_\transition \in \mathbb{Z}^Q$ the vector with 
$\character_\transition(\transition') = 1$, if $\transition' = \transition$, 
and $\character_\transition(\transition') = 0$, otherwise.
Using this notation we rewrite ($\conSysA_\transition$) to the equivalent 
constraint system ($\conSysA_\transition'$),
where $\identity$ denotes the identity matrix:

\vspace{0.2cm}
\begin{tabular}{|lc|}
	\hline
	{\begin{minipage}[r]{0.4\linewidth}
			\vspace{-0.1cm}
			constraint system ($\conSysA_\transition'$):
			\vspace{0.8cm}
		\end{minipage}}
		&
		{\begin{minipage}[r]{0.4\linewidth}
				\vspace{-0.1cm}
				\begin{eqnarray*}
					\begin{pmatrix}
						\updates \\
						\identity
					\end{pmatrix}
					\counters & \ge &
					\begin{pmatrix}
						\zeroVec \\
						\character_\transition
					\end{pmatrix}\\
					\flowMatrix \counters & = & \zeroVec
				\end{eqnarray*}
				\vspace{-0.4cm}
			\end{minipage}}\\
			\hline
		\end{tabular}
		\vspace{0.2cm}
		
		Using Farkas' Lemma (note that $\bz$ is not restricted so we can take $\bz = -\bw$), we see that either ($\conSysA_\transition'$) is 
		satisfiable or the following constraint system 
		($\conSysB_\transition'$) is satisfiable:
		
		\vspace{0.3cm}
		\begin{tabular}{|r|c|}
			\hline
			{\begin{minipage}[r]{0.49\linewidth}
					\vspace{0.2cm}
					\hspace{-0.15cm}
					constraint system ($\conSysB_\transition'$):
					\begin{eqnarray*}
						\begin{pmatrix}
							\rankCoeff \\
							\by
						\end{pmatrix} & \ge & \zeroVec \\
						\begin{pmatrix}
							\rankCoeff \\
							\by
						\end{pmatrix}^\top \cdot 
						\begin{pmatrix}
							\updates \\
							\identity
						\end{pmatrix} -						 
						\offsets^\top \cdot \flowMatrix & = & 
						\zeroVec^\top \\
						\begin{pmatrix}
							\rankCoeff \\
							\by
						\end{pmatrix}^\top\cdot
						\begin{pmatrix}
							\zeroVec \\
							\character_\transition
						\end{pmatrix} - \offsets^\top \cdot \zeroVec & > & 0
					\end{eqnarray*}
					\vspace{-0.1cm}
				\end{minipage}}
				&
				{\begin{minipage}[c]{0.45\linewidth}
						simplified version of
						constraint system ($\conSysB_\transition'$):
						\begin{eqnarray*}
							\rankCoeff & \ge & \zeroVec \\
							\by & \ge & \zeroVec \\
							\rankCoeff^\top \cdot \updates + \by^\top - \offsets^\top \cdot \flowMatrix & = & \zeroVec^\top \\
							\by(\transition) & > & 0
						\end{eqnarray*}
					\end{minipage}}\\
					\hline
				\end{tabular}
				\vspace{0.3cm}
				
				We recognize that constraint systems ($\conSysB_\transition'$) 
				and ($\conSysB_\transition$) are equivalent,
				because solutions of ($\conSysB_\transition'$) with 
				$\by(\transition) > 0$ can always be turned into solutions with 
				$\by(\transition) \ge 1$ by multiplying with a sufficiently large 
				positive rational number.

\section{Proof of Lemma~\ref{lem:goodQRF}}

We consider some vector $\bv \in \Inc$ with $\bc_f^\top \cdot \bv  = 0$.
We have $\valueSum(\cycle) = \bv$ for some simple cycle $\cycle$.
We consider the edges along the cycle $\cycle$.
Every edge of $\cycle$ must be $\qrf$-neutral:
otherwise we could add the equations $\normalVec_\qrf^\top \update + \weights_\qrf(\loc') \le \weights_\qrf(\loc)$ for every transition $(\loc,\update,\loc')$ along $\cycle$ in order to witness $\normalVec_f^\top \cdot \bv < 0$.
Hence, for every transition $\transition$ along $\cycle$ there is no other QRF $\qrf'$ such that $\transition$ is $\qrf'$-ranked;
otherwise, $\qrf$ would not be maximal with regard to the number of $\qrf$-ranked transitions.
Thus, constraint system $\conSysB_\transition$ is unsatisfiable for every transition $\transition$ of $\cycle$.
By Lemma~\ref{lem:ranking-or-witness}, constraint system $\conSysA_\transition$ is satisfiable.
For every transition $\transition$ of $\cycle$,
we fix some non-negative multi-cycle $\multCycle_\transition$ associated to some integer solution of $\conSysA_\transition$.
We take the union of the non-negative multi-cycles $\multCycle_\transition$ in order to obtain the non-negative multi-cycle $\multCycle$.
We note that $\multCycle$ contains every transition~$\transition$ of $\cycle$.
Hence, $\multCycle$ can be decomposed into the simple cycle $\cycle$ and a set of simple cycles, whose effect corresponds to some vector $\bw \in \cone(\Inc)$.
Because $\multCycle$ is non-negative, we get $\bv + \bw \geq \vec{0}$.

\section{Proof of Lemma~\ref{lem-compens}}
We need the following two lemmas.

 \begin{lemma}
 	\label{lem-kappa}
 	Let $X \subset \Z^d$ be a finite set and $\bn$ such that $X \subset
 	\hat{\calH}_{\bn}$. Then there is a constant $\kappa \in \R^+$ such that
 	for every $\bw \in \cone(X)$, where $\Norm(\bw) = 1$, there exist $k \in
 	\N$, $a_1,\ldots,a_k \in \R^+$, and $\bv_1,\ldots,\bv_k \in X$ such that
 	$\bw = \sum_{j=1}^k a_j \cdot \bv_j$ and for any
 	$S \subseteq \{1,\dots,k\}$, the absolute values of all components of the vector
 	$\sum_{j \in S} a_j \cdot \bv_j$ are bounded by~$\kappa$.
 \end{lemma}
\begin{proof}
	Let $\bw = \sum_{j=1}^k a_j \cdot \bv_j$ where $a_j \in \R^+$, $\bv_j \in 
	X$ for all $1 \leq j \leq k$, and $k$ is \emph{minimal}. First, we show 
	that for every $j \leq k$, the vector $-\bv_j$ does not belong to 
	$\cone(\{\bv_1,\ldots,\bv_{j-1},\bv_{j+1},\ldots,\bv_k\})$. Assume the 
	converse, such $j$ exists. Wlog let $j=1$, i.e., $-\bv_1 \in \cone(\{\bv_2,\ldots,\bv_k\})$. Then $-\bv_1 = 
	\sum_{j=2}^k b_j \cdot \bv_j$, where $b_j \in \R^+$ for all $2 \leq j \leq 
	k$. Further, 
	\[
	\bw \quad = \quad (a_1-c) \cdot \bv_1 
	\ + \ (a_2 - c b_2)\cdot \bv_2 \ + \ \cdots \ + \ (a_k -c b_k)\bv_k
	\]
	for every $c > 0$. Clearly, there exists $c>0$ such that at least one of 
	the coefficients $(a_1-c)$, $(a_2 - c b_2),\ldots, (a_k -c b_k)$ is zero 
	and the other remain positive, which contradicts the minimality of~$k$. 
	Since $\{\bv_1,\ldots,\bv_k\} \subseteq \hat{\calH}_{\bn}$, there must 
	exist $\bn'$ such that $\{\bv_1,\ldots,\bv_k\} \subseteq \calH_{\bn'}$ 
	(otherwise, we can use the hyperplane separation theorem argument similar 
	to the one in the proof of Case~(a) of Lemma~\ref{lem-no-halfspace} to show 
	that \mbox{$-\bv_j \in 
	\cone(\{\bv_1,\ldots,\bv_{j-1},\bv_{j+1},\ldots,\bv_k\})$} for some $1 \leq 
	j \leq k$). Since $\bv_j \cdot \bn' < 0$ for all $1 \leq j \leq k$, each 
	$\bv_j$ moves in the direction of $-\bn'$ by some fixed positive distance. 
	Since $\Norm(\bw) = 1$, there is a bound $\delta_{\bv_1,\ldots,\bv_k} \in 
	\R^+$ such that $a_j \leq \delta_{\bv_1,\ldots,\bv_k}$ for all $1 \leq j 
	\leq k$, because no $a_j\cdot \bv_j$ can go in the direction of $-\bn'$ by 
	more than a unit distance. 
	
	The above claim applies to every $\bw \in \cone(X)$. Since $X$ is finite, 
	there are only finitely many candidates for the set of vectors 
	$\{\bv_1,\ldots,\bv_k\}$ used to express $\bw$, and hence there exists a 
	fixed upper bound $\delta \in \R^+$ for all $\delta_{\bv_1,\ldots,\bv_k}$. 
	This means that, for every $\bw \in \cone(X)$, there exist $k \in \N$, 
	$a_1,\ldots,a_k \in \R^+$, and $\bv_1,\ldots,\bv_k \in X$ such that $\bw = 
	\sum_{j=1}^k a_j \cdot \bv_j$, and $a_j \leq \delta$ for all $1 \leq j \leq 
	k$. This immediately implies the existence of~$\kappa$.
\end{proof}
\begin{lemma}
Let $\A$ be a VASS and $f$ a QRF maximizing the number of $f$-ranked transitions. Then there exists $c \in \R^+$ such that for every $m \in \N$ and every cycle $\pi_m$ with $\bc_f^\top \cdot \eff(\pi_m) = 0$ and $\eff(\pi_m)\geq -\vec{m}$ there exists $\bu \in \cone(Inc)$ such that $\eff(\pi_m) + \bu \geq 0$ and there exist $a_j \in \R^+$ and $\bv_j \in \Inc$ such that
$\bu = \sum_{j=1}^k a_j \cdot \bv_j$ and for every
 	$S \subseteq \{1,\dots,k\}$, the minimum of all components of the vector
 	$\sum_{j \in S} a_j \cdot \bv_j$ is bounded by $-c \cdot m$.
\end{lemma}
\begin{proof}
From Lemma~\ref{lem:goodQRF} we obtain that for every $\bv \in \Inc$ such that $\bc_f^\top \cdot \bv = 0$
 there exists $\bw \in \cone(\Inc)$ such that $\bv + \bw \geq \vec{0}$. If $\bv \notin \Inc$ but $\bv \in \cone(\Inc)$,
  we may claim the same (since, in this case, the vector $\bv$ is just some positive combination of vectors from $\Inc$).
   We need to restrict $\Norm(\bv)$. In order to do that, consider a VASS $\A' = (Q,T')$ where
    $T' = T \cup \{(q,-e_i,q) \mid q \in Q\, 1\leq i \leq d\}$
     ($e_i$ is the $i$-th base vector i.e., $e_i(j) = 0$ for $j \neq i$ and $e_i(i) = 1$).

Let $g$ be a QRF for $\A'$ maximizing the number of $g$-ranked transitions. We want to prove that transition of $\A$ is $f$-ranked if and only if it is $g$-ranked. Surely $g$ is also QRF for $\A$ and $k \cdot f$ is QRF for $\A'$ for a suitable constant $k\in\R^+$ (it may happen that $0>\bc_f^\top \cdot e_i>-1$). But then $kf+g$ is a QRF and every $g$-ranked or $f$-ranked transition is $(kf+g)$-ranked. Therefore, every transition of $\A$ must be $g$-ranked iff it is $f$-ranked. Otherwise it is a contradiction with definition of $f$ or $g$ (maximal number of ranked transitions).

The size of any counter $j$ such that $\bc_g(j) > 0$ can be linearly bounded by the size of the initial configuration $q_0\bu_0$. Let $q\bu$ be any configuration reachable from $q_0\bv_0$. Then from the definition of QRF we have
\[
\weights_g(q_0) - \weights_g(q) + \bc_g^{\top} \cdot \bu_0 \geq \bc_g^{\top} \cdot \bv.
\]
Let $a = \weights_g(q_0) - \weights_g(q)$ and $g_{max} = \max_{1\leq i \leq d} \bc_g(i)$ and $g_{min} = \min_{1\leq i \leq d} \bc_g(i)$ and $u_{max} = \max_{1 \leq i \leq d} \bu_0(i)$. Then if $\bc_g(j) > 0$ we have
\[
\frac{a + d \cdot u_{max} \cdot g_{max}}{g_{min}} \geq \bv(i).
\]
Let $b = \frac{a + d \cdot u_{max} \cdot g_{max}}{g_{min}}$.

Now consider any cycle $\pi_m$ in $\A$ such that $\bc_f^\top \cdot \eff(\pi_m) = 0$. Then $\bc_g^{\top} \cdot \eff(\pi_m) = 0$. 
This cycle can be done also in $\A'$. In $\A'$ we can produce a cycle $\pi_{m}'$ such that, at the end, all counters are bounded by $b$ (if $\bc_g(i) = 0$ then $(q,-e_i,q)$ is $g$-neutral for all $q \in Q$ and we can take path $\pi_m$ and add these loops).

From Lemma~\ref{lem:goodQRF} there is a compensating vector for $\pi_{m}'$ i.e., $\bu \in \cone(\Inc_{\A'})$ such that $\pi_{m}'+\bu \geq 0$. Again, we can take such a path that all counters are bounded. Surely $\Norm(\bu)$ is bounded by some constant depending on $m$.

Using Lemma~\ref{lem-kappa} on $\bu$ we get the constant $c$. Now take a vector $\bu'$ for $\pi_{m}'$ such that we ommit all vectors $-e_i$, i.e., $\bu' \in \cone(\Inc(\A))$.
\end{proof}

Now use the path $\pi_{m}$. Let $\bu' = \sum_{i=1}^k a_i \bv_i$ where $\bv_i \in \Inc$. Consider a multicycle where every cycle with effect $\bv_i$ appears $\lfloor a_i \rfloor$-times. Therefore $\eff(\pi_{m}) + \sum_{i=1}^k \lfloor a_i \rfloor \bv_i$ is bounded from below by some small negative constant vector (note that the number of different effects is bounded by $d$, see the proof of Lemma~\ref{lem-kappa} and each effect is bounded by $|Q|\cdot\max{\A}$) not depending on $m$. In order to connect all cycles in this multicycle, we need to have a cycle $\eta$ going through all states of $\A$. But this can cause only a decrease of $|Q|\cdot \maxup{\A}$. This gives us a possible decrease of $|Q| \cdot \maxup{\A} \cdot (d+1)$.

\end{document}